\def\UseBibLatex{1}

\ifx\SoCG\undefined%
\documentclass[11pt]{article}%
\providecommand{\SoCGVer}[1]{}%
\providecommand{\NotSoCGVer}[1]{#1}%
\else%
\makeatletter
\def\input@path{{lipics/}{../lipics/}}
\makeatother
\documentclass[a4paper,USenglish,cleveref,autoref,thm-restate, nolineno ]{socg-lipics-v2021}
\hideLIPIcs 
\providecommand{\SoCGVer}[1]{#1}%
\providecommand{\NotSoCGVer}[1]{}%
\fi

\providecommand{\BibLatexMode}[1]{}
\providecommand{\BibTexMode}[1]{#1}

\ifx\UseBibLatex\undefined%
  \renewcommand{\BibLatexMode}[1]{}
  \renewcommand{\BibTexMode}[1]{#1}
\else
  \renewcommand{\BibLatexMode}[1]{#1}
  \renewcommand{\BibTexMode}[1]{}
\fi

\BibLatexMode{%
   \usepackage[style=alphabetic,backend=biber,bibencoding=utf8]{biblatex}%
   \usepackage{sariel_biblatex}%
}

\IfFileExists{sariel_computer.sty}{\def\sarielComp{1}}{}
\ifx\sarielComp\undefined%
\newcommand{\SarielComp}[1]{}
\newcommand{\NotSarielComp}[1]{#1}%
\else
\newcommand{\SarielComp}[1]{#1}%
\newcommand{\NotSarielComp}[1]{}%
\fi
\newcommand{\IfPrinterVer}[2]{#2}%

\NotSoCGVer{%
   \usepackage[cm]{fullpage}%
}%
\usepackage{amsmath}%
\usepackage{amssymb}%
\usepackage{xcolor}%
\NotSoCGVer{%
   \usepackage{euscript}%
}%
\usepackage{scalerel}%

\SarielComp{\usepackage{sariel_colors}}%

\NotSoCGVer{%
   \usepackage[amsmath,thmmarks]{ntheorem}%
   \theoremseparator{.}%
}%

\usepackage{titlesec}%
\titlelabel{\thetitle. }%
\usepackage{xcolor}%
\usepackage{mleftright}%
\usepackage{xspace}%
\usepackage{hyperref}%

\usepackage{wasysym}%

\IfPrinterVer{%
   \usepackage{hyperref}%
}{%
   \usepackage{hyperref}%
   \hypersetup{%
      breaklinks,%
      ocgcolorlinks, colorlinks=true,%
      urlcolor=[rgb]{0.25,0.0,0.0},%
      linkcolor=[rgb]{0.5,0.0,0.0},%
      citecolor=[rgb]{0,0.2,0.445},%
      filecolor=[rgb]{0,0,0.4},
      anchorcolor=[rgb]={0.0,0.1,0.2}%
   }
}

\DeclareFontFamily{U}{BOONDOX-calo}{\skewchar\font=45 }
\DeclareFontShape{U}{BOONDOX-calo}{m}{n}{
  <-> s*[1.05] BOONDOX-r-calo}{}
\DeclareFontShape{U}{BOONDOX-calo}{b}{n}{
  <-> s*[1.05] BOONDOX-b-calo}{}
\DeclareMathAlphabet{\mathcalb}{U}{BOONDOX-calo}{m}{n}
\SetMathAlphabet{\mathcalb}{bold}{U}{BOONDOX-calo}{b}{n}
\DeclareMathAlphabet{\mathbcalb}{U}{BOONDOX-calo}{b}{n}

\numberwithin{figure}{section}%
\numberwithin{table}{section}%
\numberwithin{equation}{section}%

\theoremseparator{.}%

\theoremstyle{plain}%
\newtheorem{theorem}{Theorem}[section]

\newtheorem{lemma}[theorem]{Lemma}

\newtheorem{fact}[theorem]{Fact}

\theoremstyle{plain}%
\theoremheaderfont{\sf} \theorembodyfont{\upshape}%
\newtheorem*{remark:unnumbered}[theorem]{Remark}%
\newtheorem{remark}[theorem]{Remark}%

\newtheorem{defn}[theorem]{Definition}

\newcommand{\myqedsymbol}{\rule{2mm}{2mm}}

\theoremheaderfont{\em}%
\theorembodyfont{\upshape}%
\theoremstyle{nonumberplain}%
\theoremseparator{}%
\theoremsymbol{\myqedsymbol}%
\newtheorem{proof}{Proof:}%

\definecolor{blue25emph}{rgb}{0, 0, 11}
\providecommand{\emphic}[2]{%
   \textcolor{blue25emph}{%
      \textbf{\emph{#1}}}%
   \index{#2}}

\providecommand{\emphi}[1]{\emphic{#1}{#1}}

\definecolor{almostblack}{rgb}{0, 0, 0.3}
\providecommand{\emphw}[1]{{\textcolor{almostblack}{\emph{#1}}}}%

\newcommand{\atgen}{\symbol{'100}}
\newcommand{\SarielThanks}[1]{\thanks{Department of Computer Science;
      University of Illinois; 201 N. Goodwin Avenue; Urbana, IL,
      61801, USA; {\tt sariel\atgen{}illinois.edu}; {\tt
         \url{http://sarielhp.org/}.} #1}}
\newcommand{\StavThanks}[1]{%
   \thanks{Department of Computer Science;
      University of Illinois; 201 N. Goodwin Avenue; Urbana, IL,
      61801, USA; {\tt stava2\atgen{}illinois.edu}; {\tt
         \url{https://publish.illinois.edu/stav-ashur}.} #1}}

\newcommand{\HLink}[2]{\hyperref[#2]{#1~\ref*{#2}}}
\newcommand{\HLinkSuffix}[3]{\hyperref[#2]{#1\ref*{#2}{#3}}}

\newcommand{\figlab}[1]{\label{fig:#1}}
\newcommand{\figref}[1]{\HLink{Figure}{fig:#1}}

\newcommand{\thmlab}[1]{{\label{theo:#1}}}
\newcommand{\thmref}[1]{\HLink{Theorem}{theo:#1}}

\providecommand{\deflab}[1]{\label{def:#1}}

\newcommand{\remlab}[1]{\label{rem:#1}}
\newcommand{\remref}[1]{\HLink{Remark}{rem:#1}}%

\newcommand{\itemlab}[1]{\label{item:#1}}
\newcommand{\itemref}[1]{\HLinkSuffix{}{item:#1}{}}

\newcommand{\lemlab}[1]{\label{lemma:#1}}
\newcommand{\lemref}[1]{\HLink{Lemma}{lemma:#1}}%

\newcommand{\seclab}[1]{\label{sec:#1}}
\newcommand{\secref}[1]{\HLink{Section}{sec:#1}}

\providecommand{\eqlab}[1]{}%
\renewcommand{\eqlab}[1]{\label{equation:#1}}

\providecommand{\remove}[1]{}%
\newcommand{\Set}[2]{\left\{ #1 \;\middle\vert\; #2 \right\}}
\newcommand{\pth}[2][\!]{\mleft({#2}\mright)}%

\newcommand{\ceil}[1]{\left\lceil {#1} \right\rceil}

\newcommand{\cardin}[1]{\left| {#1} \right|}%

\renewcommand{\th}{th\xspace}

\renewcommand{\Re}{\mathbb{R}}%
\usepackage[inline]{enumitem}

\newlist{compactenumA}{enumerate}{5}%
\setlist[compactenumA]{topsep=0pt,itemsep=-1ex,partopsep=1ex,parsep=1ex,%
   label=(\Alph*)}%

\newlist{compactenuma}{enumerate}{5}%
\setlist[compactenuma]{topsep=0pt,itemsep=-1ex,partopsep=1ex,parsep=1ex,%
   label=(\alph*)}%

\newlist{compactenumI}{enumerate}{5}%
\setlist[compactenumI]{topsep=0pt,itemsep=-1ex,partopsep=1ex,parsep=1ex,%
   label=(\Roman*)}%

\newlist{compactenumi}{enumerate}{5}%
\setlist[compactenumi]{topsep=0pt,itemsep=-1ex,partopsep=1ex,parsep=1ex,%
   label=(\roman*)}%

\newlist{compactitem}{itemize}{5}%
\setlist[compactitem]{topsep=0pt,itemsep=-1ex,partopsep=1ex,parsep=1ex,%
   label=\bullet}%

\newcommand{\etal}{\textit{et~al.}\xspace}

\newcommand{\Arr}{\Mh{\mathop{\mathrm{\EuScript{A}}}}}%

\newcommand{\Term}[1]{\textsf{#1}}

\newcommand{\BSP}{\Term{BSP}\xspace}%
\newcommand{\VC}{\Term{VC}\xspace}%
\newcommand{\LP}{\Term{LP}\xspace}%
\newcommand{\ULP}{\Term{ULP}\xspace}%

\newcommand{\eps}{\varepsilon}

\providecommand{\Mh}[1]{#1}%

\newcommand{\Undecided}{Undecided\xspace}

\newcommand{\pp}{\Mh{p}}%
\newcommand{\pq}{\Mh{q}}%
\newcommand{\pd}{\Mh{\mathcalb{d}}}%
\newcommand{\pc}{\Mh{c}}%
\newcommand{\PS}{\Mh{P}}%
\newcommand{\PSA}{\Mh{T}}%

\newcommand{\LS}{\Mh{L}}%
\newcommand{\QS}{\Mh{Q}}%
\newcommand{\BSet}{\Mh{B}}%

\newcommand{\hplane}{\Mh{h}}%
\newcommand{\hsp}{\Mh{\alpha}}%
\newcommand{\hspA}{\Mh{\beta}}%

\newcommand{\HSet}{\Mh{H}}%

\newcommand{\RangeSpace} {\Mh{\mathsf{S}}}
\newcommand{\RangeSet}{{\Mh{\mathcal{R}}}}

\newcommand{\SetC}{\Mh{N}}%

\newcommand{\PgSet}{\Mh{\mathcal{D}}}%
\newcommand{\Polygon}{\Mh{D}}%
\newcommand{\EdgesX}[1]{\Mh{E}\pth{#1}}%
\newcommand{\ArrX}[1]{\Arr\pth{#1}}%
\newcommand{\VVX}[1]{\Mh{\mathcal{V}}\pth{#1}}%
\newcommand{\VV}{\Mh{\mathcal{V}}}%

\newcommand{\interiorX}[1]{\mathop{\mathrm{int}}\pth{#1}}%

\newcommand{\Body}{\Mh{\mathcal{C}}}%

\newcommand{\PT}{\Mh{\mathcal{K}}}%
\newcommand{\PTX}[1]{\Mh{\EuScript{I}}\pth{#1}}%

\newcommand{\tldO}{\scalerel*{\widetilde{O}}{j^2}}%

\newcommand{\CHX}[1]{\Mh{\mathcal{CH}}\pth{#1}}

\newcommand{\CSet}{\Mh{\mathcal{C}}}%
\newcommand{\Cutting}{\Mh{\Xi}}%

\newcommand{\Simplex}{\Mh{\nabla}}

\newcommand{\Line}{\Mh{\ell}}

\newcommand{\seg}{\Mh{s}}%
\newcommand{\kopt}{\Mh{\ensuremath{\mathcalb{k}}}\xspace}%
\newcommand{\NN}{\textsf{NN}\xspace}%
\newcommand{\FN}{\textsf{FN}\xspace}%

\newcommand{\polylog}{\mathrm{polylog}}%

\newcommand{\diskY}[2]{\Mh{\mathrm{disk}}\pth{#1, #2}}%
\newcommand{\mDiskY}[1]{\Mh{\mathrm{\vartheta}}\pth{#1}}%
\newcommand{\mToPlaneX}[1]{\Mh{\mathrm{\varphi}}\pth{#1}}%

\newcommand{\ball}{\Mh{\mathcalb{b}}}
\newcommand{\qball}{\Mh{\mathcalb{q}}}

\newcommand{\GroundSet}{\Mh{\textsf{X}}}%

\newcommand{\FGroundSet}{\Mh{\mathsf{x}}}%

\newcommand{\MeasureChar}{\overline{m}}

\newcommand{\range}{\Mh{\mathbf{r}}}

\newcommand{\MeasureX}[1]{\MeasureChar\pth{#1}}

\newcommand{\sMeasureX}[1]{\overline{s}\pth{#1}}
\newcommand{\sMeasureY}[2]{\overline{s}^{}_{#2}\pth{#1}}

\newcommand{\I}{\Mh{\mathcal{I}}}%

\newcommand{\Dim}{\Mh{\delta}}%

\newcommand{\BadProb}{\Mh{\varphi}}%
\providecommand{\si}[1]{#1}

\newcommand{\Family}{\Mh{\mathcal{F}}}%
\newcommand{\FamilyX}[1]{\Mh{\mathcal{F}}\pth{#1}}%

\newcommand{\Sample}{\Mh{R}}%
\newcommand{\sSize}{\Mh{\nu}}%

\newcommand{\Turan}{Tur\'{a}n\xspace}%

\newcommand{\naively}{na\"{i}vely\xspace}

\newcommand{\permut}[1]{\left\langle {#1} \right\rangle}
\newcommand{\avgX}[1]{\mathrm{avg}\pth{#1}}%

\usepackage{stmaryrd}%
\providecommand{\IntRange}[1]{\mleft\llbracket #1 \mright\rrbracket}
\newcommand{\IRX}[1]{\IntRange{#1}}%

\newcommand{\epsH}{\Mh{\tau}}%
\newcommand{\prjX}[1]{#1_{\downarrow}}

\newcommand{\TriSet}{\Mh{\Xi}}

\BibLatexMode{\bibliography{undecided}}

\title{On Undecided LP, Clustering and Active Learning}

\NotSoCGVer{%
   \author{%
      Stav Ashur%
      \StavThanks{}%
      \and%
      Sariel Har-Peled%
      \SarielThanks{Work on this paper was partially supported by a
         NSF AF award CCF-1907400.  }%
   }%
}

\SoCGVer{%
   \author{Stav Ashur}%
   {Department of Computer Science, University of Illinois, 201
      N. Goodwin Avenue, Urbana, IL 61801, USA}%
   {stava2@illinois.edu}%
   {https://orcid.org/0000-0003-0533-8978}%
   {}%
   \author{Sariel Har-Peled}%
   {Department of Computer Science, University of Illinois, 201
      N. Goodwin Avenue, Urbana, IL 61801, USA}%
   {sariel@illinois.edu}%
   {https://orcid.org/0000-0003-2638-9635}%
   {Work on this paper was partially supported by a NSF AF award
      CCF-1907400.}%
}%

\SoCGVer{%
   \authorrunning{S. Ashur and S. Har-Peled} %
   \Copyright{Stav Ashur and Sariel Har-Peled}%
   \ccsdesc[500]{Theory of computation~Computational geometry}%
   \keywords{Linear Programming, Active learning, Classification} }

\NotSoCGVer{\date{\today}}

\begin{document}

\maketitle

\begin{abstract}
    We study colored coverage and clustering problems. Here, we are
    given a colored point set where the points are covered by
    (unknown) $\kopt$ clusters, which are monochromatic (i.e., all the
    points covered by the same cluster, have the same color). The
    access the colors of the points (or even the points themselves) is
    provided indirectly via various queries (such as nearest neighbor,
    or separation queries). We show that if the number of clusters is
    a constant, one can correctly deduce the color of all the points
    (i.e., compute a monochromatic clustering of the points) using a
    polylogarithmic number of queries.

    We investigate several variants of this problem, including
    \emph{Undecided Linear Programming}, covering of points by $\kopt$
    monochromatic balls, covering by $\kopt$ triangles/simplices, and
    terrain simplification. For the later problem, we present the
    first near linear time approximation algorithm. While our
    approximation is slightly worse than previous work, this is the
    first algorithm to have subquadratic complexity if the terrain has
    ``small'' complexity.
\end{abstract}

\section{Introduction}

Given a set of points $\PS$ in $\Re^d$ that are labeled (say, colored
as red and blue), the problem of learning a classifier that labels the
points correctly is a standard machine learning question. In the
active learning settings, querying/exposing the label of an input is
an expensive endeavor, and one tries to minimize such queries while
performing the learning task.

We are interested in a somewhat related question: If the input point
set has a ``simple'' structure, but we are given access to the input
via oracles that performs more ``interesting'' queries than just
exposing the label of a point, can one classify correctly all the
input points using relatively few oracle queries?

\paragraph*{Implicit input model.}

Consider the situation where instead of the algorithm reading the
input as in the classical setting, the access to the input is via
\emph{input primitives}, or \emph{oracles}. Such indirect access to
the data rises naturally if the data is already stored in a
preexisting database or data-structure. This approach is of relevance
nowadays as large amount of data makes even the basic task of reading
the whole input infeasible or prohibitively expensive.

This gives rise to the main motivation for this work -- what input
primitives/oracles one needs, so that one can derive efficient
algorithms. Here, of special interest are algorithms with running
times that are sublinear in the input size.

\paragraph*{Problem I: \Undecided linear programs.}
An instance of linear programming is a set of $n$ linear inequalities
on $d$ variables, where one needs to find an assignment of real values
to the variables, such that all the inequalities hold.  We consider a
new variant of \LP, first studied by Maass and \Turan
\cite{mt-clcmq-90}, where the $n$ linear constraints are given, but we
do not know a priori whether the inequality is $\leq$ or $\geq$ for
each one of them. Geometrically, this corresponds to being given $n$
hyperplanes in $\Re^d$, each having two closed halfspaces associated
with it. Which of the two halfspaces is the one used in the \LP can be
revealed by querying an oracle. For example, the ``standard''
\emph{separation oracle}, which returns for a given query point $\pp$,
a violated constraint of the \LP, or alternatively returns that $\pp$
is feasible. An example of \ULP (i.e., undecided \LP) is depicted in
\figref{undecided_LP_example}.

\begin{figure}[ht]
    \begin{tabular}{cc}
      \begin{minipage}{0.46\linewidth}
          \smallskip%
          \centering \includegraphics[width=0.9\textwidth, page=3]%
          {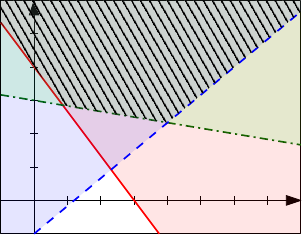}%

          \smallskip%
      \end{minipage}
      &
        \begin{minipage}{0.46\linewidth}
            \centering \includegraphics[width=0.6\textwidth, page=4]%
            {figs/undecided_l_p}%

        \end{minipage}
      \\
      (i) A set of ``undecided'' constraints in 2d...
      &%
        (ii) ... is  a set of lines in the plane.
      \\[0.3cm]
      \begin{minipage}{0.46\linewidth}
          \centering \includegraphics[width=0.9\textwidth, page=1]%
          {figs/undecided_l_p}%
      \end{minipage}
      &
        \begin{minipage}{0.46\linewidth}
            \centering \includegraphics[width=0.9\textwidth, page=5]%
            {figs/undecided_l_p}%
        \end{minipage}
      \\
      \begin{minipage}{0.46\linewidth}
          \medskip
          
          \textbf{(iii)} A possible commitment of the constraints, and
          the feasible polygon induced.  \smallskip
      \end{minipage}
      &
        \begin{minipage}{0.46\linewidth}
            \medskip
            
            \textbf{(iv)} An alternative commitment of the underlying
            constraints, and the induced polygon.
        \end{minipage}
    \end{tabular}
    
    \caption{An instance of 3 undecided constraints (bottom) with two
       possible sets of underlying decided constraints (bottom left and
       right).}
    \figlab{undecided_LP_example}
\end{figure}

\paragraph*{Problem II: Separating red and blue points, with a %
   counterexample oracle.}
The above problem, in the dual, is the following: The input is a set
of unlabeled points, and the task is to compute a hyperplane
separating the blue points from the red points. The ``separation''
oracle here, is given an oriented hyperplane that is supposed to
separate the points, and the oracle returns a misclassified point.  A
natural question is how many queries of this type one has to perform
until classifying all the points correctly.

\paragraph*{The learning model.}  Our model seems to be Angluin's
equivalence query model for active learning \cite{a-qcl-87}. In
particular, Maass and \Turan \cite{mt-hwtgl-94, mt-albolgc-94} studied
the above two problems.  See \remref{m:t} for more details.

\paragraph*{Problem III: Covering/clustering points with a ball %
   using proximity oracle.}

Consider the situation where the input is a set of colored points in
$\Re^d$, where all the (say) red points are inside a ball, and all the
points outside the ball are blue.  Here, our access to the color of
the points is via an oracle that can answer colored nearest-neighbor
(\NN) or furthest-neighbor (\FN) queries (that is, the oracle can
return the closest red [or blue] point to the query point). The task
at hand is to label (i.e., color) the points correctly using a minimal
number of oracle queries.

\paragraph*{The challenge.}
To appreciate the difficulty in solving the above problem, consider
the natural naive algorithm -- pick a random sample, expose the colors
of the points in the sample (in this case the colored \NN queries can
do that), compute a ball separating the red points and blue points in
the sample, and feed it into the ``counterexample'' oracle (which
returns a colored point that is on the wrong side of the ball -- this
oracle can be implemented using the \NN/\FN queries). The algorithm
adds this counterexample to the current set of points whose color is
known. Now the algorithm repeats the process finding an updated
separating circle for the points whose colors are known.  This
algorithm does arrive to the right answer, but it is easy to come up
with examples where it has to do a linear number of iterations. As
such, the challenge is to get a sublinear number of iterations.

\paragraph*{Problem IV: Covering points by monochromatic balls.}

Consider the situation where the input is a set of points that can be
covered by $\kopt$ balls (i.e., clusters) such that all the points
covered by the same ball have the same label/color (i.e., red or
blue). In this setting, given a query point, the oracle returns the
closest point of a prespecified color.  See \figref{k_disks_example}
for an example.

\begin{figure}[ht]
    \centering
    \begin{minipage}{.48\textwidth}
        \centerline{%
           \includegraphics[width=\textwidth, page=1]{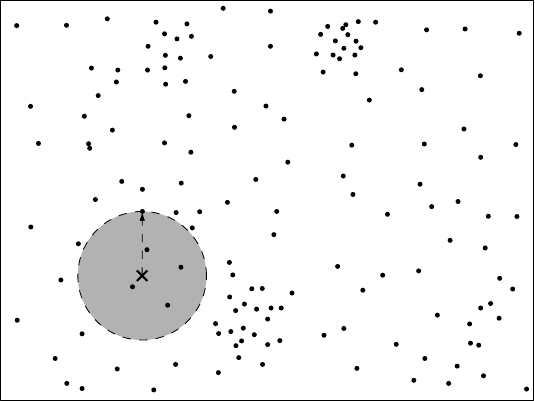}%
        }%

        \smallskip%
        (A) \NN(blue) oracle query marked by an ``x'' reveals that the
        points in the interior of the disk are red, and the returned
        point is blue.
    \end{minipage}%
    \hfill
    \begin{minipage}{.48\textwidth}
        \centerline{%
           \includegraphics[width=\textwidth,page=6]{figs/k_disks}%
        }

        \smallskip%
        (B) {The same query with the underlying colors of the points,
           and an optimal solution with $\kopt=12$ monochromatic
           disks.}
    \end{minipage}
    \caption{An instance of Problem IV.}
    \figlab{k_disks_example}
\end{figure}

\paragraph*{Problem V: Covering points with monochromatic triangles.}

The input is a set of unlabeled points that can be covered by $k$
triangles. All the points covered by the same triangle, have the same
label/color (i.e., red or blue). The oracle, given a triangle and a
color, returns either that (i) all the points covered by the triangle
have the same color, or (ii) a point covered by the triangle is of the
opposite color.

\paragraph*{Problem VI: Terrain simplification.}

Here, the input is a set of points in three dimensions sampled from a
terrain. For each sample point we have an associated interval in the
$z$-axis of length 2$\eps$ with the point in its center, that
represents the tolerated approximation error.  The problem is then to
compute a terrain made out of a minimal number of triangles, that
passes through all the vertical error tolerance intervals of the
points. This problem can be interpreted as a set cover problem with
planar triangles, where a triangle is allowable only if it can be
lifted to three dimensions so that it approximates the points lying in
it correctly.

\begin{figure}
    \centering%
    \begin{tabular}{|c|c|c|c|}
      \hline
      dim
      & ref
      & RT
      & \# queries
      \\
      \hline\hline                             
      $d=2$
      &
        \thmref{ulp:2:fast}%
      &
        $O(n )$ & $O( \log n) \Bigr.$ %
      \\
      \hline
      \hline
      $d=3$
      &
        \lemref{ulp:3:d:3/2}%
      &
        $\tldO(n^{3/2})$ & $O( \log n) \Bigr.$%
      \\
      \hline
      &
        \thmref{ulp:3:d}%
      &
        $\tldO(n^{1+\delta})$ & $O( \delta^{-1} \log n) \Bigr.$   \\
      \hline
      \hline
      $d$
      &
        \begin{minipage}{0.18\linewidth}
            \centering \smallskip

            \cite{mt-hwtgl-94, mt-albolgc-94}
            
            \lemref{ulp:h:d}%

            \smallskip
        \end{minipage}
      &
        $n^{O(d^2)}$ & $O(d^2 \log n)\Bigr.$ %
      \\
      \hline
      &
        \lemref{ulp:h:d}%
      & $\tldO\pth{n^{d}} $ & $O(d^3 \log n)\Bigr.$ \\
      \hline
      &
        \lemref{cutting}%
      & $O_d(n)$ & $O_d( \log^d n)\Bigr.$ \\
      \hline
      &
        \thmref{u:d:lp:cutting}%
      & $O_d(n )$ & $O_d( \log^{d-1} n)\Bigr.$ \\
      \hline
      \hline
      $d > 3$
      &
        \remref{ulp:d:b:3}%
      & $\tldO\pth{n^{1+\delta}} $ & $O_d( \log^{d-2} n)\Bigr.$
      \\
      \hline
      \hline
    \end{tabular}
    \caption{A summary of the results on \ULP. Here $\delta$ can be
       chosen to be any constant in $(0,1)$. The $O_d$ hides constants
       that depends (probably exponentially or worse) on the
       dimension.}
\end{figure}

\paragraph*{Related work.}

Linear programming has a long history, see the survey by
\cite{dmw-lp-04}.

In \emphw{active learning}, also known as query learning or
experimental design, the purpose of the algorithm is learning a
concept by querying specific input entries for their label, and the
main criteria for efficiency is minimizing the number of queries.  The
basic premise is that asking a specialist to label a specific example
is an expensive operation. See \cite{s-alls-09} for a survey on the
topic of active learning.

Closer to our settings, Har-Peled \etal \cite{hjr-alcbl-20}, studied
algorithms for actively learning a convex body using a separation
oracle. Such an oracle either confirms that that the query point is
within the convex body, or alternatively, returns a hyperplane
separating the point and the convex region.

Kane \etal \cite{klmz-accq-17} studied half plane classifiers using
comparison queries, and showed an exponential query complexity
improvement over learning with only membership queries. Their model is
somewhat similar to the model of Problem II above (of separating red
and blue points), except that their model assumes that the oracle can
return the distance of a query point to the optimal separating
hyperplane, while our model only assumes that the oracle can identify
a misclassified point.

In general, computational models that involve various oracles as
algorithmic building blocks have been studied in computational
geometry, as they represent algorithms in which the input is given
implicitly, and access to any information provided by the input is
done by oracle queries. See Har-Peled \etal \cite{hkmr-seps-16} and
references therein for such examples.

As mentioned above, Maass and \Turan \cite{mt-hwtgl-94, mt-albolgc-94}
studied Problems I and II -- our results are better, but only in
constant dimensions, see \remref{m:t} for details.

\begin{figure}
    \centering
    \begin{tabular}{*{3}{|c}|l|}
      \hline
      Cover size
      &
        Terrain size
      & Running time
      & ref\\
      \hline\hline
      $O(\kopt \log \kopt)$
      &
        $O(\kopt \log \kopt)$
      &
        $O(n^8) \Bigr.$
      &
        Agarwal and Suri \cite{as-sagp-98}     \\       %
      \hline
      $O(\kopt \log \kopt)$
      &
        $O(\kopt^2 \log^2 \kopt)$
      &
        $O(n^{2+\eps} + \kopt^3 \polylog )\Bigr.$
      &
        Agarwal and Desikan \cite{ad-eats-97}\\
      \hline
      $O(\kopt \log n)$
      &
        $O(\kopt^2 \log^2 n)$
      &
        $O\pth{n  \cdot \kopt^{14} \log^{13} \kopt \log n}\Bigr.$
      &
        \thmref{t:simp}\\
      \hline
    \end{tabular}
    \caption{Results on terrain simplification.}
    \figlab{terrains}
\end{figure}

\paragraph{Previous work on terrain simplification.}
The input is a set $\PS$ of $n$ points in three dimensions, and a
vertical error parameter $\epsH>0$. The purpose is to return a terrain
of minimum complexity, such that each point of $\PS$ is in vertical
distance at most $\epsH$ from the terrain. In the following, we assume
that there is such a terrain made out of $\kopt$ triangles.

As mentioned above, one can interpret this problem as a hitting set or
a set cover problem. Here, the task is to compute a collection of
triangles, such that for any point of $\PS$, the vertical line through
it intersects at least one triangle, and all the triangles it
intersects are in vertical distance at most $\epsH$ from the point.

Agarwal and Suri \cite{as-sagp-98} presented an algorithm that
computes, in roughly $O(n^8)$ time, a terrain approximation with
$O( \kopt \log \kopt )$ triangles -- their basic scheme is based on
recovering via dynamic programming a \BSP partition of the optimal
vertical decomposition.  A faster algorithm was presented by Agarwal
and Desikan \cite{ad-eats-97}, with running time
$O(n^{2+\eps} + \kopt^3 \polylog )$, that outputs a set of
$O( \kopt \log \kopt )$ triangles that approximates the point set.
Converting this set of triangles into terrain, results in a simplified
terrain of size $O( \kopt^2 \log^2 \kopt )$.  See \figref{terrains}
for summary of these results.

\subsection{Our results}

\begin{compactenumA}
    \item \textsc{Undecided \LP.}  We present several algorithms for
    solving undecided \LP{}s. In \secref{naive:u:l:p}, we revisit the
    algorithm of Maass and \Turan{}, showing \ULP{}s can be solved
    using $O(d^2 \log n)$ oracle queries. The main tool is repeatedly
    computing a centerpoint and feeding it to the oracle to further
    truncate the search space. This bound is polynomial in $d$, but
    the runtime of algorithm itself is doubly exponential in the
    dimension $d$. The running time can be improved to (roughly)
    $O(n^d)$, with the number of queries deteriorating to
    $O(d^3 \log n)$.

    In \secref{u:l:p:cutting} we present a linear time algorithm for
    constant dimension that uses cutting, but the number of separation
    oracle queries is now $O( \log^d n)$.  In \secref{u:l:p:plane} we
    show that in the plane, in near linear time, one can reduce the
    number of queries to $O( \log n)$.

    \medskip
    \item \textsc{Optimal algorithm for Undecided LP in the plane}.
    In \secref{u:l:p:2:d:linear:time}, we present the ultimate
    algorithm in the plane, that runs in linear time, and performs
    only logarithmic number of queries. This requires a careful
    combination of the previous algorithm, together with random
    sampling of the constraints, and an iterative refinement
    algorithm. This result is one of the highlights of this paper.

    \medskip%
    
    \item \textsc{Covering (red) points by a single ball.}
    In \secref{single:ball} we study Problem III, and present an
    algorithm that uses $O( \log^2 n )$ colored \NN/\FN queries, and
    computes the single ball that covers (say) all the red points, and
    avoids all the blue points. The algorithm works by lifting the
    input point set to three dimensions, and then using the algorithm
    for undecided \LP.

    \smallskip%
    \item \textsc{Covering points by $\kopt$ monochromatic balls.}  In
    \secref{multi:ball} we address Problem IV above, where the input
    is covered by \kopt monochromatic balls, and we have access to a
    colored \NN oracle. Inspired by the one ball case, we show a
    greedy algorithm that finds a ball that covers $O(1/\kopt)$
    faction of the uncovered points. This leads to an algorithm that
    performs $O(\kopt^{d+2} \log^{d+2} n )$ queries (see
    \thmref{main}) and correctly classifies all the points.

    \smallskip%
    \item \textsc{Covering points by $\kopt$ monochromatic triangles.}
    We are given (i) a colored set $\PS$ of $n$ points in the plane,
    (ii) an oracle that can report, given a triangle $\triangle$,
    whether or not the color of all points inside $\triangle \cap \PS$
    are the same (and if so, report the color), and (iii) an oracle,
    such that given a partial cover by triangles, randomly returns an
    uncovered point.  We present an algorithm that computes a cover of
    $\PS$ by $O(\kopt \log n)$ monochromatic triangles, under the
    assumption that such a cover exists, in $k^{O(1)} \log n $ time,
    see \thmref{cover:by:triangles} for the exact bounds. This
    algorithm is significantly sublinear if $k \ll n$. If we are not
    given the oracles, but instead are provided with the colors of the
    points, the oracles can be implemented, so that the resulting
    algorithm works in $O(n k^{O(1)} \log n )$ time.
    
    \smallskip%
    \item \textsc{Terrain simplification.}  The above algorithm
    applies verbatim for the problem of terrain simplification,
    implying a $O(n \kopt^{O(1)} \log n )$ time algorithm for a cover
    of the given point set (up to the error parameter provided) by
    $O(\kopt \log n)$ triangles. This can be converted into a
    simplified terrain of complexity $O(\kopt^2 \log^2 n)$.  All
    previous work \cite{as-sagp-98, ad-eats-97} had quadratic or
    (much) worse running time.  See \secref{terrain} for details.
\end{compactenumA}

\paragraph{Technical contribution.}

For the undecided \LP our algorithms require ways to sample vertices
of an arrangement of lines/planes/hyperplanes that lies inside a
polytope (that has relatively small complexity). In two dimensions we
do so by using known techniques. To get a linear time algorithm
requires a nontrivial combination of random sampling together with
gradation type refinement approach.  In three dimensions we use a
rather interesting over-sampling idea that has the potential to be
useful in other settings. Even in three dimensions the problem of how
to sample such a vertex efficiently (i.e., linear time) remains an
interesting open problem.

Another contribution of this work is a refinement of the notion of
canonical ranges.

\begin{figure}
    \noindent{}
    \phantom{}%
    \hfill
    \includegraphics[page=3]{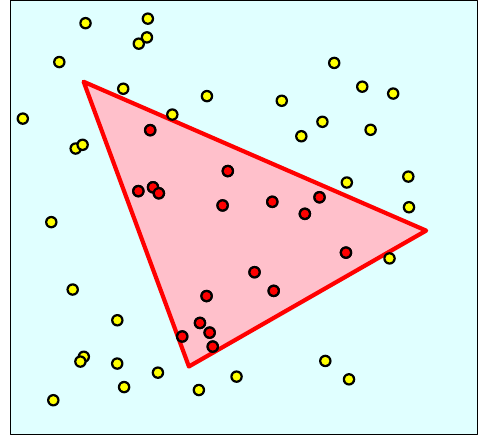}
    \hfill
    \includegraphics[page=4]{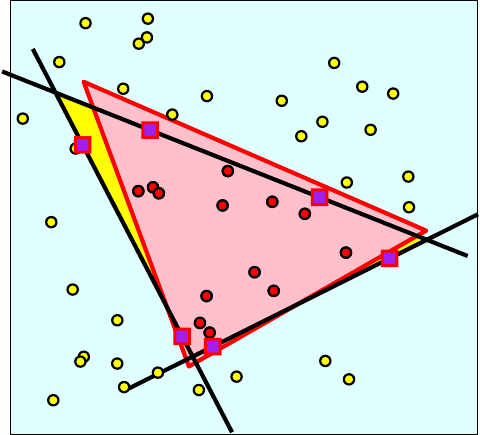}
    \hfill%
    \phantom{}%
    \caption{A triangle covering only red points. The natural
       canonical triangle, that is induced by the subset of the points
       inside the original triangle, is not contained inside the
       original triangle (see yellow regions in the right pane).}
    \figlab{canonical}
\end{figure}

\paragraph{Beyond canonical sets.}

Underlying this work is the dichotomy between ranges, the set of input
points they contain (i.e., the projection of the ranges to the set of
points), and the ranges that are defined by the original set of
points.  As a concrete example, consider a set $\PS$ of points in the
plane. A triangle $\triangle$ contains the subset of the points
$\PS \cap \triangle$. Our purpose is to recover $\triangle$, or at
least a triangle that contains the same subset of points of $\PS$.
For example, assume that all the points of $\PS \cap \triangle$ are
red, and the points of $\PS \setminus \triangle$ are blue, and we have
access to the points and their colors via some oracle.

The standard approach is via \emph{canonical triangles} -- the idea
being that one takes the original triangle, and morph it continuously
into a new triangle, so that it contains the same subset of points
(except maybe for a few points on the boundary). The new triangle
(i.e., the canonical triangle corresponding to the original triangle)
is now fully constrained/defined by the points on the boundary. Under
general position assumption, such a triangle would have six points on
its boundary, which implies that the number of different subsets of
$\PS$ that can be realized by triangles is $O(n^6)$.

The problem, for our purposes, is that the resulting canonical
triangle is not contained inside the original triangle.  This issue is
illustrated in \figref{canonical}.  This issue becomes critical for
our purposes, as we can not afford to enumerate all canonical sets
directly, because their number is prohibitive. Instead, we are
interested in the triangles induced by a (small) subset/sample of the
original points. To have any hope to get a query efficient algorithm,
we need a way for the sample to induce triangles that are contained
inside the original (unknown!) triangle, while containing a large
fraction of the sample points.

To this end, we define a new process of defining alternative triangles
-- which we refer to as \emph{self-defined} triangles.  The newly
defined triangle have the desired property that they are strictly
contained inside the unknown original triangle, while containing a
significant fraction of the points contained inside it.  See
\lemref{2k-gon} and \lemref{1:10} for details.

This issue might seem as tedium, but it is the technical insight
needed to ``unlock'' random sampling and to be able to use it for
terrain approximation. This results in the first near linear time
algorithm for terrain simplification, as all previous work required
quadratic (or worse) running time. We believe our approach should be
useful for other problems.

The idea of covering canonical sets by a few sets, that belong to a
special family of sets (which is much smaller than the original family
of canonical sets) is by now standard -- see for example the work of
Aronov \etal \cite{aes-ssena-10}. Despite the superficial similarly,
this is distinct from the self-defined sets/triangles we introduce.

\section{Algorithms for solving undecided linear %
   programs}
\seclab{undecided:L:P}

\begin{remark}
    All the algorithms described below for undecided \LP work in the
    same fashion -- they generate a sequence of separation oracle
    queries. Specifically, if any of the query points are feasible,
    the algorithm immediately stops and outputs that the \ULP is
    feasible with the queried point as a proof. For the clarity of
    description in the following, we always assume a separation oracle
    returns a separating hyperplane.
\end{remark}

\subsection{Centerpoint based algorithm for solving %
   \ULP{}s}
\seclab{naive:u:l:p}

We review the algorithm of Maass and \Turan \cite{mt-hwtgl-94,
   mt-albolgc-94} for solving undecided \LP (they did not provide
running time bounds for their algorithm, which we do).

Observe that for a set of $d$-dimensional (closed) hyperplanes
$\HSet$, if there exists a feasible point, then, under general
position assumption, there exists a vertex of the arrangement
$\Arr(H)$ that is feasible.

Let $\PS$ be a set of $n$ points in $\Re^d$.  For a parameter
$\alpha \in (0,1)$, a point $\pc \in \Re^d$ is an
\emphi{$\alpha$-centerpoint} if all halfspaces containing $\pc$ also
contain at least $\alpha n$ points of $\PS$.  A classical implication
of Helly's theorem, is that for any set $\PS$ of $n$ points in
$\Re^d$, there is a $1/(d+1)$-centerpoint. Such a point is simply a
\emphi{centerpoint} of $\PS$, and can be computed in $O(n^{d-1})$ time
\cite{jm-ccfps-94,c-oramt-04}.  A $2/d^2$-centerpoint can be computed
in near linear time \cite{hj-jcps-19} (here, the running time is
polynomial in $d$).

\paragraph*{The algorithm.}  %
The input is a set $\HSet$ of $n$ hyperplanes in $d$ dimensions.  The
first step of the algorithm is to compute the set $\PS$ of vertices of
the arrangement $\Arr(H)$. The number of such vertices is
$\leq \binom{n}{d}$. As long as $\PS$ has more than $d^2\log n$
points, the algorithm computes a centerpoint $\pc$ of $\PS$. The
algorithm then queries the separation oracle on $\pc$ in order to
decide whether $\pc$ is feasible. If it is, then the algorithm is
done, as it computed a feasible point. Otherwise, the oracle returned
a violated constraint of the given \LP (this also provides the
algorithm with the direction of the constraint for the associated
input hyperplane). The algorithm removes all the points of $\PS$ that
violate this constraint, and repeats until $|\PS| = O(d^2 \log
n)$. Once $\PS$ is that small, the algorithm simply checks the
feasibility of each of the remaining vertices by querying it with the
separation oracle.

\begin{lemma}
    \lemlab{ulp:h:d}%
    The above algorithm computes a feasible point of $\HSet$ using
    $O( d^2 \log n)$ separation oracle queries. The running time of
    this algorithm, ignoring the oracle calls, is $O( n^{d(d-1)})$
    time.

    Alternatively, the algorithm can be modified so that it computes a
    feasible point using $O( d^3 \log n)$ separation oracle
    queries. The running time of this algorithm, ignoring the oracle
    calls, is $\tldO( n^{d})$ time, where $\tldO$ hides
    polylogarithmic terms.

\end{lemma}
\begin{proof}
    Let $m_0 = \binom{n}{d}$ be the number of vertices in $\Arr(H)$.
    Every centerpoint computation and the following oracle call
    reduces the number of vertices by a factor of $1-1/(d+1)$.  As
    such, after $O( (d+1) \log_2 \binom{n}{d}) = O( d^2 \log n )$
    iterations the algorithm is done.

    Computing all of the vertices of $\Arr(H)$ takes $O(n^d)$
    time. Computing a centerpoint takes $O(n^{d(d-1)})$ time
    \cite{jm-ccfps-94,c-oramt-04}. Repeating this $d+1$ times, reduces
    the number of active points by a constant factor, and the overall
    running time is thus bounded by
    $\sum_{i=0}^{\lg n} O(d (n/2^i)^{d(d-1)}) = O(d n^{d(d-1)})$.

    As for the modified algorithm, let $\PS$ be the set of vertices of
    the arrangement $\Arr(H)$. Let $m_0 = \cardin{\PS}$.  In the
    $i$\th iteration, we use the algorithm of Har-Peled and Jones
    \cite{hj-jcps-19} that computes a $O(1/d^2)$-centerpoint of $\PS$
    in $d^{O(1)} \log^6 n$ time, with high probability. We now use
    this centerpoint in the above algorithm instead of the exact
    $1/(d+1)$-centerpoint. The bounds stated readily follow, as it
    takes $O(d^2)$ iterations to halve the number of active points,
    and the bottleneck in the running time of the algorithm is
    computing $\PS$ initially.
\end{proof}

\subsection{Cutting based algorithm}
\seclab{u:l:p:cutting}

Here, we present the new algorithm to solve undecided \LP using
cuttings.

\begin{defn}
    For a set of $n$ $d$-dimensional hyperplanes $\HSet$, a
    \emphi{${1}/{r}$-cutting} of $\HSet$ is a partition $\Cutting$ of
    $\Re^d$ into $O(r^d)$ simplices, such that the interior of each
    simplex of the cutting intersects at most $n/r$ hyperplanes of
    $\HSet$.  The list of hyperplanes intersecting the interior of a
    simplex $\Simplex \in \Cutting$, is the \emphi{conflict list} of
    $\Simplex$.
\end{defn}

A $1/r$-cutting, and its conflict list can be computed in
$O( nr^{d-1})$ time \cite{h-gaa-11}.

\begin{remark}
    \remlab{oracle:all}%
    The algorithm we present next call itself recursively on subsets
    of the constraints, and on lower dimensional subspaces.  In
    particular, the oracle can be applied to any lower dimensional
    affine subspace $F$ by using the original oracle in the ambient
    space -- a returned constraint can be intersected with $F$ to get
    a constraint in $F$. We emphasize that the oracle always works on
    the whole original input set of constraints -- the recursive calls
    on subsets of the constraints are done for efficient bookkeeping,
    and do not effect how the oracle works.
\end{remark}

\paragraph*{The Algorithm.}

The algorithm computes a $1/r$-cutting of $\HSet$. Let $\Cutting$ be
this set of simplices, and let $\Cutting_{d-1}(\Cutting)$ be the set
of $O(r^d)$ $(d-1)$-dimensional simplices that form the faces of the
simplices of $\Cutting$. The algorithm now solves the problem
recursively on each of the $(d-1)$-dimensional simplices of
$\Cutting_{d-1}(\Cutting)$, and the hyperplanes of $\HSet$ that
intersects it.  Each recursive call is on a problem that is one
dimensional lower, and involves only $n/r$ constraints. The oracle
still applies to the whole set of constraints, see
\remref{oracle:all}.

If any of these recursive calls finds a feasible point, then we are
done. Otherwise, the recursive calls performed involved calls to the
oracle, and forced some of the constraints to expose themselves. Let
$\CSet$ be the set of these committed halfspaces (i.e., all the
halfspaces returned by the separation oracle). If the intersection of
all these constraints is empty (i.e., the associated \LP is
infeasible), then this can be discovered, in $O(|\CSet|)$ time, by
invoking a standard \LP solver on $\CSet$. If the \LP is feasible,
then it returns us a point $\pp$ inside the polytope
\begin{equation*}
    \PT = \bigcap_{h^+ \in \CSet} h^+.
\end{equation*}
Furthermore, this polytope must be fully contained in the interior of
one of the simplices of $\Cutting$ (otherwise, the algorithm would
have found a feasible point in one of the recursive calls). By
scanning $\Cutting$, we discover the simplex $\Simplex \in \Cutting$
that contains $\pp$. The algorithm now call recursively on $\Simplex$
and its conflict list (passing $\CSet$ as the current set of committed
constraints).

\begin{figure}[ht]
    \begin{tabular}{cc}
      \includegraphics[page=1]%
      {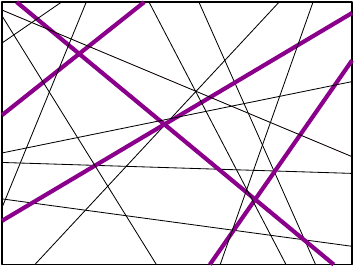}
      & \includegraphics[page=2]{figs/cutting_algorithm}\\
      \begin{minipage}{0.44\linewidth}
          Cells of a cutting of an arrangement of undecided
          constraints.
      \end{minipage}
      &
        \quad
        \quad
        \begin{minipage}{0.44\linewidth}
            A single cell of the cutting might contain a possibly
            feasible region in its interior.
        \end{minipage}
    \end{tabular}
    \caption{Illustration of an iteration of the cutting-based
       algorithm. A cutting of an arrangement of undecided constraints
       is depicted on the left.  After solving the problem recursively
       on the cuttings of the edges, some of the constraints are now
       committed. The feasible region (the red polygon) induced by the
       constraints is a polygon that must be fully contained in one of
       the cells of the cutting.}
    \figlab{cuttingAlg}
\end{figure}

\paragraph*{Algorithm in one dimension.}

The 1-dimensional case is solved using a binary search. We have $n$
uncommitted rays on the real line, and our purpose is to find an
atomic interval that is feasible. To this end, the algorithm computes
a median among the points defined by the rays, and then asks the
oracle to commit the direction of the appropriate ray. This decreases
the number of potential atomic intervals that might be feasible by
half. At the end of the process, the algorithm is left with a single
atomic interval that might be feasible -- the algorithm asks the
oracle whether the middle of this interval is feasible or not.

As such, after $O( \log n)$ iterations, the algorithm is done, as each
iteration either finds a feasible point, or throws away half of the
rays.  The running time is $O(n)$, and the number of oracle queries
performed is $O( \log n)$.

\paragraph*{Analysis: Number of oracle queries.}  The query complexity
of this algorithm is
\begin{equation*}
    Q_d(n)%
    =%
    O(r^d)Q_{d-1}\pth{\frac{n}{r}} + Q_d\pth{\frac{n}{r}},
\end{equation*}
with $Q_1(n) = O( \log n)$. The solution of this recurrence is
$O( \log^{d} n)$, for $r$ chosen to be a sufficiently large
constant. Indeed, using induction, we have
\begin{math}
    Q_d(n)%
    =%
    O(r^d \log^{d-1} \frac{n}{r} ) + Q_d\pth{\frac{n}{r}},
\end{math}

\paragraph*{Running time.}  As for the running time, we have
\begin{equation*}
    T_d(n)%
    =%
    O(nr^{d-1}) + O(r^d) T_{d-1}(n/r) + T_d(n/r),
\end{equation*}
with $T_1(n) =O(n)$. Assuming $T_{d-1}(n) \leq c_{d-1} n$, and for two
constants $c_d'$ and $c_d''$, we have
\begin{equation*}
    T_d(n)%
    \leq%
    c_d' nr^{d-1} + c_d'' r^d c_{d-1} \frac{n}{r}
    + c_d\frac{n}{r}
    \leq%
    \Bigl(c_d' r^{d-1} + c_d'' r^{d-1} c_{d-1}  + \frac{c_d}{r}\Bigr) n
    \leq 
    c_d n,
\end{equation*}
which holds if
\begin{math}
    c_d \geq 2(c_d' r^{d-1} + c_d'' r^{d-1} c_{d-1} ).
\end{math}
This implies that $T_d(n) = O(n)$. We thus get the following result.

\begin{lemma}
    \lemlab{cutting}%
    Undecided \LP with $n$ constraints in $\Re^d$, can be solved in
    $O_d(n)$ time, using $O( \log^d n)$ separation oracle queries.
\end{lemma}

\begin{remark}[An implicit undecided \LP]
    \remlab{implicit}%
    In the following, we need a variant of the above problem -- the
    input is a set of undecided constraints, but the real undecided
    \LP instance $\I$ corresponds to an (unknown) subset of these
    constraints. Fortunately, the given separation oracle works on the
    ``real'' instance of constraints $\I$. It is easy to verify that
    the above algorithm of \lemref{cutting} works verbatim in this
    case.
\end{remark}

\begin{remark}
    \remlab{m:t}%
    The work of Maass and \Turan \cite{mt-hwtgl-94, mt-albolgc-94}
    also studied the dual settings of our problem (i.e., Problem II of
    separating red and blue points). They assume the input points are
    taken from a grid of bounded spread, and use the ellipsoid
    algorithm to get an efficient algorithm with a small number of
    queries.
\end{remark}

\subsection{Near linear time algorithm in two dimensions}
\seclab{u:l:p:plane}

\subsubsection{\VC dimension of polygons}

\begin{lemma}
    \lemlab{arr_poly_vc_dim}%
    Let $\LS$ be a set of $n$ lines in $\Re^2$, and let $\PgSet^k$ be
    the set of all convex polygons that have a boundary made out of at
    most $k$ edges or rays.  For a polygon $\Polygon \in \PgSet^k$,
    let
    \begin{math}
        \Polygon \sqcap \LS%
        =%
        \Set{ \hplane \in \LS}{\hplane \cap \interiorX{\Polygon} \neq
           \emptyset },
    \end{math}
    and consider the range space
    \begin{math}
        \RangeSpace^k%
        =%
        (\LS, \Set{\Polygon \sqcap \LS}{\Polygon \in \PgSet^k} ).
    \end{math}
    The \VC dimension of $\RangeSpace^k$ is $O( k \log k)$.
\end{lemma}

\begin{proof}
    The \VC dimension of
    \begin{math}
        \RangeSpace_{\seg}%
        =%
        (\HSet, \Set{\seg \sqcap \HSet}{\seg \text{ is a segment}})
    \end{math}
    is $O(1)$, being the dual of the range space of points, with the
    ranges being double-wedges. For a polygon $\Polygon \in \PgSet^k$,
    let $\EdgesX{\Polygon}$ be the set of segments and rays forming
    its boundary.  Observe that $\hplane \in \HSet$ intersects the
    interior of $\Polygon$ $\iff$ $\hplane$ intersects the interior of
    one of the segments/rays of $\EdgesX{\Polygon}$. Namely, a range
    of $\RangeSpace^k$ is the union of at most $k$ ranges of
    $\RangeSpace_{\seg}$. By standard bounds on the \VC dimension of
    combined range spaces \cite{h-gaa-11}, we have that the \VC
    dimension of $\RangeSpace^k$ is bounded by $O( k \log k)$.
\end{proof}

We need the following lemma, due to Har-Peled and Mitchell
\cite{hj-ospl-20}.

\begin{lemma}[\cite{hj-ospl-20}]
    \lemlab{count:and:sample}%
    Let $\Polygon$ be a fixed polygon with $k$ edges, and let $L$ be a
    set of $m$ lines that intersect the interior of $\Polygon$. After
    $O(m (\log k + \log m))$ preprocessing, one can compute the number
    of vertices of $\Arr(L)$ that lie inside $\Polygon$, or sample
    such a vertex in $O( \log m)$ time.
\end{lemma}

\subsubsection{Polygon and point set reduction}

\paragraph{Polygon reduction.}
\begin{lemma}
    \lemlab{polygon:reduction}%
    Consider an instance of undecided \LP in the plane, and a polygon
    $\Polygon$ with $t$ edges, such that the feasible region is
    contained in $\Polygon$. Using $O( \log t)$ separation queries,
    and $O(t)$ time, one can either compute a feasible point, or
    compute a polygon $\Polygon' \subseteq \Polygon$ with (say) at
    most $10$ edges, such that the feasible region must be contained
    in $\Polygon'$.
\end{lemma}

\begin{proof}
    Let $\Polygon_1 =\Polygon$ and $t_1 = t$.  As long as the current
    polygon $\Polygon_i$ has more than $10$ edges, compute a
    centerpoint $\pp_i$ to the vertices of $\Polygon_i$. Next, using a
    separation oracle query, decide if $\pp_i$ is feasible. If it is,
    the algorithm is done. Otherwise, the oracle returns a halfplane
    $\hplane^+$ that must contain the feasible solution, such that
    $\pp_i \notin \hplane^+$. This implies that
    $\Polygon_{i+1} \leftarrow \Polygon_i \cap \hplane^+$ has at most
    $t_{i+1} \leq (2/3)t_i + 2$ vertices. The algorithm continues to
    the next iteration.

    Clearly, after $O( \log t)$ iterations, the current polygon has
    less than $10$ vertices, and the algorithm returns it as the
    desired polygon $\Polygon'$. The running time follows as the
    centerpoint can be computed in the plane in linear time, and the
    number of vertices of the polygon reduces by a constant fraction
    at each iteration.
\end{proof}

\paragraph{Point set reduction via centerpoint.}
\begin{lemma}
    \lemlab{p:s:reduction}%
    Consider an instance of undecided \LP in $\Re^d$, and a set $\PS$
    of $n$ points. In $O(n)$ time, using $O(\log n)$ separation oracle
    queries, one can compute either: (i) a feasible point for the
    \ULP, or alternatively, (ii) a polytope $\Polygon$ with
    $O(d^2 \log n)$ faces, such that the feasible solution for the
    \ULP lies inside $\Polygon$, and $\Polygon$ contains no point of
    $\PS$.

    In two dimensions, one can modify the algorithm so that the number
    of edges of $\Polygon$ is at most $10$.
\end{lemma}
\begin{proof}
    We apply the same centerpoint reduction idea used above. Let
    $\PS_1 = \PS$. In the beginning of the $i$\th iteration, the
    algorithm computes a $O(1/d^2)$-centerpoint $\pp_i$ for $\PS_i$ by
    using the algorithm of Har-Peled and Jones \cite{hj-jcps-19},
    which works in $d^{O(1)} \log^6 n$ time (and succeeds with high
    probability).
 
    Next, using a separation oracle query, the algorithm decides if
    $\pp_i$ is feasible. If so, the algorithm is done. Otherwise, the
    algorithm removes all the points of $\PS_i$ that are on the same
    side of the returned hyperplane (or line in 2d) as $\pp_i$, sets
    $\PS_{i+1}$ to be the resulting point set, and continues to the
    next iteration. Since
    $|\PS_{i+1}| \leq (1-O(1/d^2))|\PS_{i}|+O(d)$, the algorithm ends
    up with a set containing a constant number of points after
    $O( d^2 \log n)$ iterations. The algorithm uses the separation
    oracle on each of he remaining points to eliminate them. Finally,
    the returned polytope is the intersection of all the constraints
    returned by the separation oracle. Since the algorithm uses
    $O( 1 + \log n)$ separation queries, the bound on the complexity
    of the output polytope follows.

    To get the improved algorithm in the plane, the algorithm of
    \lemref{polygon:reduction} is executed on the resulting polygon to
    further reduce the number of its edges to $10$.
\end{proof}

\subsubsection{Near linear running time in two dimensions}

Let $\LS$ be the input set of $n$ lines (i.e., undecided constraints).
Let $\LS_0 = \LS$, and $\Polygon_0 = \Re^2$.

Let $V_i$ denote the set of vertices of $\ArrX{\LS_{i-1}}$ that lie in
$\Polygon_{i-1}$.  The algorithm computes $n_i = |V_i|$ using the
algorithm of \lemref{count:and:sample}. There are three possibilities:
\begin{compactenumA}
    \medskip%
    \item \itemlab{step:empty} If $n_i = 0$, then the algorithm picks
    any point $\pp_i$ in $\Polygon_{i-1}$, and calls the separation
    oracle on $\pp_i$. If $\pp_i$ is feasible, the algorithm is done,
    otherwise, the algorithm returns that the given instance is
    infeasible.

    \medskip%
    \item If $n_i = O( n \log n)$, the algorithm computes $V_i$ by
    clipping the lines of $\LS$ to $\Polygon_{i-1}$, and computing the
    arrangement of the resulting segments.  This takes
    $O( n \log n+ n_i)$ expected time, as this is the time required
    for computing the arrangement of $n$ segments with $n_i$
    intersections.

    The algorithm uses \lemref{p:s:reduction} on $V_i$ to either find
    a feasible point, or a polygon $\Polygon_1'$ that must contain the
    feasible region, has at most $10$ edges, and contains no point of
    $V_i$.

    \medskip%
    \item Otherwise, the algorithm samples $O( n )$ vertices from
    $V_i$, and sets $\Sample_i$ to be the resulting sample. This is
    done using the algorithm of \lemref{count:and:sample} in
    $O(n \log n)$ time. Next, the algorithm applies, as above,
    \lemref{p:s:reduction} and \lemref{polygon:reduction} to get a
    constant complexity polygon $\Polygon_i \subseteq \Polygon_{i-1}$
    that does not contain any of the points of $\Sample_i$.
\end{compactenumA}
\smallskip%
The algorithm now scans the lines of $\LS_{i-1}$, computes the set of
all lines $\LS_{i} \subseteq \LS_{i-1}$ that intersect $\Polygon_i$,
and continues to the next iteration.

\begin{lemma}
    \lemlab{ulp:2:d}%
    The above algorithm solves two dimensional undecided \LP in
    expected $O(n \log n)$ time, using $O( \log n)$ separation oracle
    queries.
\end{lemma}
\begin{proof}
    All the edges of the current polygon $\Polygon_i$ are contained in
    lines of $\LS$. As such, when $\Polygon_i$ contains no vertex of
    $\ArrX{\LS}$, then it must be a face of the arrangement, and thus
    \itemref{step:empty} applies, and the result it reports is
    correct.

    Now, $\Sample_i$ is a sample of $O( n)$ vertices from
    $V_i$. Constant complexity convex polygons have constant \VC
    dimension by \lemref{arr_poly_vc_dim}, As such, $\Sample_i$ is an
    $\eps$-net for $V_i$, where $\eps = O( (\log n ) / n)$. Since
    $\Polygon_i$ has constant complexity and does not contain any
    point of $V_i$, it follows that it contains at most $\eps |V_i|$
    vertices of $\ArrX{\LS}$. We conclude that
    $n_{i+1} \leq \eps n_i $. Clearly, after three such iterations,
    the algorithm has no vertices left, and the algorithm is done.

    Finally, since each iteration takes $O(n \log n)$ expected time,
    the claim follows.
\end{proof}

\subsection{A linear time algorithm with logarithmic %
   number of queries in two dimensions}

\seclab{u:l:p:2:d:linear:time}

Let $\LS$ be a set of $n$ undecided constraints in the plane (i.e.,
lines) as before.

\paragraph*{The algorithm.}

The algorithm samples a set $\Sample$ of $\Theta(n / \log n)$ lines
from $\LS$, by picking each line with probability $p = c/\log n$ for
some appropriate constant $c$. The algorithm then computes a polygon
$\Polygon_0$ with at most $10$ edges that contains no vertex of
$\ArrX{\Sample}$ by using the algorithm of \lemref{ulp:2:d} on
$\Sample$. Note that the oracle queries are still done on the original
set of constraints $\LS$.

Let $\LS_0$ be the set of all the lines of $\LS$ that intersect the
interior of $\Polygon_0$.  Formally, we have
\begin{equation*}
    \LS_0%
    =%
    \LS \sqcap \Polygon_0%
    =%
    \Set{\Line \in
       \LS}{ \Line \cap \mathrm{int}(\Polygon_0) \neq \emptyset \bigr.}.
\end{equation*}
Similarly, let $\Sample_0 = \Sample \sqcap \Polygon_0$.

\begin{figure}[ht]
    \begin{tabular}{*{3}{c}}
      \includegraphics[page=1]{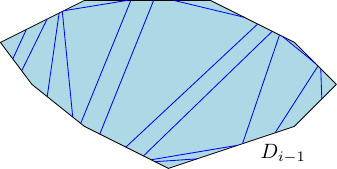}
      &
        \includegraphics[page=4]{figs/cut_up}%
      &%
        \includegraphics[page=5]{figs/cut_up}%
      \\
      (A) & (B) & (C)
    \end{tabular}
    \noindent%
    \centerline{%
       \begin{tabular}{cc}
         \includegraphics[page=6]{figs/cut_up}
         &
           \includegraphics[page=7]{figs/cut_up}
         \\
         (D) & (E)
       \end{tabular}%
    }
    \caption{An illustration of the steps composing an iteration of
       the algorithm described in \secref{u:l:p:2:d:linear:time}. (A)
       Compute a polygon with no arrangement vertices (B) find a
       separating cell $\Polygon_i'$ of the subpolygons created by
       $\Sample_{i-1}$ (C) compute the intersection points $\PS_i$ of
       the lines $\LS_{i-1}$ with $\Polygon_i'$ (D) compute
       $\Polygon_i''$ using \lemref{p:s:reduction} (E) get
       $\Polygon_i$ from the intersection of $\Polygon_i''$ and
       $\Polygon_{i-1}$.}
    \figlab{alg:steps}
\end{figure}

For $i > 0$, in the $i$\th iteration the algorithm computes the
segments formed by the intersections of the lines of $\Sample_{i-1}$
with $\Polygon_{i-1}$. These segments are interior disjoint, and they
partition $\Polygon_{i-1}$ into $m_i + 1$ subpolygons, where
$m_{i-1} =\cardin{\Sample_{i-1}}$, see \figref{alg:steps}. The dual of
the subpolygons is a tree, and there is a vertex which, when removed,
breaks the polygon into parts such that each part contains at most
$(2/3) (m_{i-1} + 1)$ original subpolygons -- this separating vertex
can be computed in linear time in the tree size. This vertex
corresponds to a subpolygon $\Polygon_i'$ that has $O(1)$ edges. The
algorithm computes all the intersections of the lines of $\LS_{i-1}$
with the boundary of $\Polygon_i'$, and sets $\PS_i$ to be the
resulting set of points. The algorithm then uses
\lemref{p:s:reduction} on $\PS_i$ in order to get a constant
complexity polygon $\Polygon_i''$ that contains the feasible region,
and does not contain any point of $\PS_i$. This requires
$O( \log |\PS_i|)$ queries. (If $\Polygon_i''$ does not intersect
$\Polygon_{i-1}$, then the given instance is infeasible, and the
algorithm can stop.)

The boundary of the polygon $\Polygon_i''$ does not intersect the
boundary of the polygon $\Polygon_i'$. As such, $\Polygon_i''$ either
\begin{compactenumi}[leftmargin=2cm]
    \smallskip%
    \item intersects only one of the parts of
    $\Polygon_{i-1} \setminus \Polygon_i'$, and the algorithm sets
    this part to be $\Polygon_i$, or \smallskip%
    \item $\Polygon_{i}'' \subset \Polygon_i'$, and the algorithm then
    sets $\Polygon_i = \Polygon_i'$.
\end{compactenumi}
\smallskip%
The algorithm next computes $\LS_i = \Polygon_i \sqcap \LS_{i-1}$ and
$\Sample_i = \Polygon_i \cap \Sample_{i-1}$ and continues to the next
iteration.

As soon as $\cardin{\LS_i} = O( n / \log n)$, the algorithm stop the
iterations, and calls the algorithm of \lemref{ulp:2:d} on $\LS_i$
(and $\Polygon_i$ as the region containing the feasible solution), in
order to find the feasible solution for the original instance, if it
exists.

\subsubsection{Analysis}

A polygon $\widehat{\Polygon}$ is \emphi{supported} by $\LS$, if
$\widehat{\Polygon}$ is the union of some faces of the arrangement
$\ArrX{\LS}$. All the polygons computed by the above algorithms are
supported by $\LS$.

\begin{lemma}
    \lemlab{at:most:20}%
    For all $i$, the polygons $\Polygon_i$ and $\Polygon_i'$ have at
    most $20$ vertices.
\end{lemma}
\begin{proof}
    Observe that $\Polygon_0$ has at most $10$ edges, as it was
    computed by the algorithm of \lemref{ulp:2:d}.  The polygon
    $\Polygon_i$ is the result of taking $\Polygon_0$, partitioning it
    along some of the lines of $\Sample$, and picking one of the
    pieces. By construction, the arrangement $\ArrX{ \Sample}$ has no
    vertex inside $\Polygon_0$.

    An edge of $\Polygon_i$ is either a subsegment of an original edge
    of $\Polygon_0$, or alternatively, it connects two points that lie
    on two original edges of $\Polygon_0$. As the edges of
    $\Polygon_i$ must alternate between original edges and new edges,
    it follows that $\Polygon_i$ has at most $20$ edges/vertices.

    The same argument applies verbatim to $\Polygon_i'$.
\end{proof}

\begin{lemma}
    The boundaries of the polygons $\Polygon_i'$ and $\Polygon_i''$ do
    not intersect.
\end{lemma}

\begin{proof}
    As noted above, both $\Polygon_i'$ and $\Polygon_i''$ are
    supported by $\LS$.  In particular, all their vertices are
    vertices of $\ArrX{\LS}$. As such, $\PS_i$ is the set of all
    vertices of $\ArrX{\LS}$ that lie on $\partial \Polygon_i'$. If a
    feasible point was found during the execution of the algorithm
    used to compute $\Polygon_i''$, then the algorithm would have
    stopped. As such, all the oracle queries returned (closed)
    halfspaces (defined by lines of $\LS$), and the intersection of
    these half-spaces (i.e., $\Polygon_i''$) does not contain any
    vertex of $\PS_i$.  If the boundaries of $\Polygon_i'$ and
    $\Polygon_i''$ intersect, then this intersection would contain a
    point of $\PS_i$, which is a contradiction.
\end{proof}

The following is an immediate consequence of Chernoff's inequality,
and we omit the easy proof.

\begin{lemma}
    We have $pn/2 \leq m_0 = \cardin{\Sample} \leq 2pn = O(n /\log n)$
    with high probability.
\end{lemma}

\begin{lemma}
    \lemlab{drop}%
    The polygon $\Polygon_i$ intersects at most
    $n_i = \cardin{\LS_i} = O( (2/3)^in )$ lines of $\LS$ with high
    probability, as long as $(2/3)^i \geq 1/n^{1/3}$.  In particular,
    the algorithm performs $O( \log \log n)$ iterations with high
    probability.
\end{lemma}

\begin{proof}
    We have that $m_0 = \cardin{\Sample}$, and
    \begin{math}
        m_i%
        =%
        \cardin{ \Sample \sqcap \Polygon_i}%
        =%
        \cardin{ \Sample_i \sqcap \Polygon_i}%
        \leq%
        (2/3)m_{i-1} + O(1).
    \end{math}
    As such, we have $m_i = O((2/3)^i n / \log n)$, as
    $m_0 = O(n /\log n)$.  Consider an $\eps$-approximation for $\LS$,
    where the ranges are polygons with at most $20$ vertices. By
    \lemref{arr_poly_vc_dim}, this range space has \VC dimension
    $O(1)$, and by setting
    \begin{equation*}
        \eps = \frac{c' \log n}{\sqrt{n}},
    \end{equation*}
    we get that the sample $\Sample$ is an $\eps$-approximation of
    size $O( \eps^{-2} \log \eps^{-1} ) = O( n/ \log n)$, and this
    holds with high probability.  As $n = \cardin{\LS}$, we have
    \begin{equation*}
        \cardin{ \frac{\cardin{\LS_i}}{\cardin{\LS}} -
           \frac{\cardin{\Sample_i}}{\cardin{\Sample}} }
        \leq%
        \eps
        \implies%
        \frac{\cardin{\LS_i}}{n} \leq        
        \frac{ \cardin{\Sample_i}}{{\cardin{\Sample}}} + \eps %
        \implies%
        \cardin{\LS_i} \leq        
        \frac{ n  \cardin{\Sample_i}}{\cardin{\Sample}}  + \eps n 
        \implies%
        \cardin{\LS_i} \leq        
        \frac{ n  \cardin{\Sample_i}}{p n /2} + \eps n 
        =%
        O( \cardin{\Sample_i} /  p ) + \eps n.
    \end{equation*}    
    This implies that
    \begin{equation*}
        n_i%
        =%
        \cardin{\LS_i}%
        \leq%
        \eps n + O( m_i /p)
        =%
        O( \sqrt{n} \log n + (2/3)^i n)
        =%
        O( (2/3)^i n).
    \end{equation*}    
    
    The second claim follows, as the algorithm stops as soon as
    \begin{math}
        n_i%
        =%
        O((2/3)^i n) =%
        O(n/ \log n),%
    \end{math}
    for $i = O( \log \log n)$.
\end{proof}

\begin{lemma}
    With high probability, for all $i=1, \ldots, h$, we have
    $\cardin{\PS_i} = O( \log^2 n)$. As such, the number of oracle
    queries performed in each iteration is $O( \log \log n)$.
\end{lemma}
\begin{proof}
    Consider a segment $\seg$ that does not intersect any line of
    $\Sample$. The sample $\Sample$ can be interpreted as an
    $\eps$-net for the ground set of lines $\LS$, where the ranges are
    segments, and $\eps =O\bigl( (\log^2 n) / n\bigr)$, as
    $O(\eps^{-1} \log \eps^{-1}) = O(n/ \log n)$.

    The interior of the polygon $\Polygon_i'$ intersects no line of
    $\Sample$ (i.e., $\Sample_{i-1}$). This polygon has at most $20$
    edges by \lemref{at:most:20}. The $\eps$-net property implies that
    every edge of $\Polygon_i'$ intersects at most $\eps n$ lines of
    $\LS$ (infinitesimally move the edge inward so that it no longer
    lies on a line of $\Sample$).  As such, the interior of
    $\Polygon_i'$, can intersect at most $O(20 \eps n) = O( \log^2 n)$
    lines of $\LS$. Furthermore, for all $i$, we have that
    $\cardin{\PS_i} = O( \log^2 n)$, as each segment that intersects
    the interior of $\Polygon_i'$ contributes two vertices to $\PS_i$.

    As for the second part, the algorithm of \lemref{p:s:reduction}
    performs $O( \log \cardin{\PS_i} ) = O( \log \log n)$ oracle
    queries on $\PS_i$, as claimed.
\end{proof}

\begin{lemma}
    The running time of the above algorithm is $O(n)$ with high
    probability, and performs $O( \log \log n)$ iterations.
\end{lemma}
\begin{proof}
    We have that $m_0 = O( n /\log n)$. As such, the two invocations
    of \lemref{ulp:2:d} takes $O(n )$ time.

    The $(i+1)$\th iteration of the algorithm takes
    $R_i = O( m_{i} \log m_{i} + n_{i} + \cardin{\PS_i} ) = O( m_i
    \log m_i + n_i )$ time since $\Polygon_i$ has constant complexity,
    and thus computing $\LS_i$ takes linear time in
    $\cardin{\LS_{i-1}}$.

    By \lemref{drop}, the algorithm needs to perform
    $h = O( \log \log n)$ iterations till $n_i = O(n /\log n)$.  Since
    $m_i \leq (2/3)m_{i-1} + O(1)$, and $n_i = O((2/3)^i n)$, it
    follows that
    $\sum_{i=1}^h O(R_{i}) = O( m_0 \log m_0 + n) = O(n)$, as claimed.
\end{proof}

\begin{theorem}
    \thmlab{ulp:2:fast}%
    Undecided \LP with $n$ constraints in $\Re^2$, can be solved in
    $O(n)$ expected time, using $O( \log n)$ separation oracle
    queries.  These guarantees hold with high probability.
\end{theorem}
\begin{proof}
    The only missing part is bounding the overall number of queries --
    the two invocations of the algorithm of \lemref{ulp:2:d} require
    $O( \log n)$ oracle queries. The $O( \log \log n)$ iterations
    require $O( \log \log n)$ oracle queries each. Putting the two
    together we get that the overall number of queries is
    $O( \log n + (\log \log n)^2 ) = O( \log n)$.~
\end{proof}

Combining the above algorithm with the cutting based algorithm of
\lemref{cutting}, results in the following improved theorem.

\begin{theorem}
    \thmlab{u:d:lp:cutting}%
    Undecided \LP with $n$ constraints in $\Re^d$, for $d\geq 2$, can
    be solved in $O_d(n)$ time, using $O( \log^{d-1} n)$ separation
    oracle queries.  These guarantees hold with high probability.
\end{theorem}

\section{A query efficient algorithm for \ULP in three %
   dimensions}

\subsection{Emulating the two dimensional algorithm}

In three dimensions, one can still get $O( \log n)$ queries, albeit
with running time $\tldO(n ^{3/2})$.

The input is a set $\HSet$ of $n$ planes in three dimensions. The
algorithm randomly samples a set $V_1$ of
\begin{equation*}
    T = \Theta(n^{3/2} \log^3 n \sqrt{ \log \log n})    
\end{equation*}
vertices of $\ArrX{\HSet}$. Each vertex is generated by randomly
choosing three constraints (planes) from $\HSet$, and computing their
intersection.  Using the algorithm of \lemref{p:s:reduction}, one
computes a convex polytope $\PT_1$ with $O( \log n)$ faces, which does
not contain any vertex of $V_1$ (since \lemref{p:s:reduction} returns
the $O(\log n)$ halfspaces whose intersection forms the desired
polytope, the polytope itself can be computed in
$O( \log n \log \log n)$ time from the returned planes).

Next, the algorithm computes all the vertices of $\ArrX{\HSet}$ inside
$\PT_1$ -- this can be done in an output sensitive fashion. To this
end, one ``walks'' around the arrangement of $\ArrX{\HSet}$, starting
(say) with the bottom vertex of $\PT_1$. Specifically, one walks on
edges of the arrangement (inside $\PT_1$) using the data-structure of
Chan~\cite{c-dgdss-19}, which provide dynamic maintenance of
convex-hull in three dimensions, and extreme point queries. Here the
data-structure is used in the dual settings, where it maintains
dynamically a set of planes and answers ray shooting queries.  Each
operation takes $O( \log^4 n)$ amortized time. The exact details of
this exploration are somewhat delicate, but straightforward, and we
omit them as they are similar in nature to the 2d algorithms (see
\cite{h-twpa-00} and references therein).  Let $V_2$ be the resulting
set of vertices.

Now, invoking (again) the algorithm of \lemref{p:s:reduction}, one can
compute a polytope $\PT_2$ that contains no vertex of $V_2$. The
algorithm computes the polytope $\PT_1 \cap \PT_2$ (in polylogarithmic
time). If the intersection is empty then the given instance is
infeasible. Otherwise, the algorithm query a point inside this
intersection using the separation oracle. If the point is feasible
then the algorithm is done, and otherwise the instance is infeasible.

\begin{lemma}
    \lemlab{ulp:3:d:3/2}%
    The above algorithm solves Undecided \LP with $n$ constraints in
    $\Re^3$, using $O( \log n)$ separation oracle queries, in
    \begin{math}
        O(T) = O\bigl( n^{3/2} \log^3 n \sqrt{\log \log n} \bigr)
    \end{math}
    time.
\end{lemma}

\begin{proof}
    Let
    \begin{math}
        \eps%
        =%
        \Omega\bigl( { \sqrt{\log \log n}} / (n^{3/2} \log n) \bigr).
    \end{math}
    The argument of \lemref{arr_poly_vc_dim}, implies that the \VC
    dimension of polytopes in $\Re^3$ with $k$ faces is
    $O( k \log k)$. In our case, $k= O( \log n)$, and as $V_1$ is a
    random sample of size $T$, which can be interpreted as an
    $\eps$-net, as
    \begin{equation*}
        \Theta\Bigl( \frac{k \log k}{\eps} \log \frac{1}{\eps} \Bigr)%
        =%
        \Theta\Bigl( \frac{\log n \log \log n}{\eps} \log n \Bigr)%
        =%
        \Theta\Bigl( n^{3/2} \log^3 n \sqrt{\log \log n} \Bigr)%
        =%
        \Theta( T) 
        . 
    \end{equation*}
    Thus, with high probability, by the $\eps$-net theorem, $\PT_1$
    (which avoids all the vertices of $V_1$), contains at most
    $\eps n^3$ vertices. That is, $|V_2| \leq \eps n^3$. Computing
    $V_2$ thus takes
    \begin{math}
        O(|V_2| \log^4 n )%
        =%
        O( \eps n^{3} \log^4 n)%
        =%
        O(T)
    \end{math}
    time.

    Since the algorithm invokes \lemref{p:s:reduction} twice, the
    number of oracle queries it performs is $O( \log n)$.
\end{proof}

\subsection{A faster algorithm}

In the following, we assume that the set of planes of $\HSet$ is in
general position -- no three planes pass through a common line, and no
four planes have a common intersection point.

\paragraph*{Idea.} A natural approach for getting a faster algorithm
is to maintain a polytope $\PT_{i-1}$, sample an $\eps$-net for the
vertices of the arrangement inside $\PT_{i-1}$ (for an $\eps$ to yet
be specified), and use the algorithm of \lemref{polygon:reduction} to
find a low complexity polytope that avoids all the points in this
$\eps$-net. Intersecting this polytope with the previous active
polytope, results in a shrunken feasible region $\PT_i$. Furthermore,
$\PT_i$ contains an $\eps$-fraction of the vertices of the arrangement
inside it compared to $\PT_{i-1}$.  The algorithm then continues to
the next iteration, till the polytope contains no vertices of the
arrangement, and then a single query in its interior settles the
feasibility of the given \ULP.

The challenge is that despite $\PT_i$ being simple (i.e., having few
faces), we do not know how to sample uniformly and efficiently from
$\VV \cap \PT_i$, where $\VV = \VVX{\HSet}$ is the set of vertices of
$\ArrX{\HSet}$.  Instead, we offer the following over-sampling
approach.

\begin{lemma}
    \lemlab{over:sample}%
    Let $\PT$ be a polytope in three dimensions with $k$ faces, and
    let $\HSet$ be a set of $n$ planes, where all the faces of $\PT$
    lie on planes of $\HSet$. One can sample a \emph{non-empty} set
    $X$ of at most $n-1$ vertices, such that (i)
    $X \subseteq \VV \cap \PT$, where $\VV=\VVX{\HSet}$, and the
    probability of any vertex of $\VV \cap \PT$ to be included in the
    sample is the same.

    The preprocessing time of the algorithm is $O(n k \log n )$, and a
    sampled set can be computed in $O(n)$ time.
\end{lemma}

\newcommand{\face}{\Mh{f}}%

\begin{proof}
    A \emph{line} of the arrangement $\Arr = \ArrX{\HSet}$ is the
    intersection of two planes of $\HSet$. The idea is to randomly and
    uniformly pick a line of $\Arr$ that intersects $\PT$, and add all
    the vertices along this line that are in $\PT$, to the set $X$
    (observe that two intersections of this line with the boundary of
    polytope are in $X$). If there are $t$ lines of $\Arr$ that
    intersect $\PT$, then the probability of a vertex of
    $\VV \cap \PT$ to be picked is exactly $3/t$.

    Consider a 2d face $\face$ of $\partial\PT$. It forms a convex
    polygon in the plane $\hplane \in \HSet$ that supports it. Every
    line of $\Arr$ that intersects $\face$ and the interior of $\PT$
    does it in a vertex in the interior of $\face$ (i.e., formed by
    the intersection of the lines and the plane supporting
    $\face$). The idea is to pick such a vertex uniformly at random.

    To this end, the algorithm computes the number vertices of $\Arr$
    in the interior of $\face$ in $O( n \log n)$ time using
    \lemref{count:and:sample}.  As such, in $O(n k \log n )$ time, one
    can compute the number of vertices of $\Arr$ that lie in the
    interior of the faces of $\PT$. In addition, every face has the
    set of all planes that intersect its interior, which defines a set
    of lines, and can also be computed in $O(n k \log n)$ time
    overall, for all the faces of the polytope. This leaves the $O(k)$
    lines supporting the edges of the polytope, which can be computed
    as its own set explicitly. We thus have $O(k)$ disjoint sets of
    lines of the arrangement, that cover all the lines intersecting
    $\PT$ (some of these sets are implicit), and furthermore, we know
    the size of each set. We now randomly choose one of the sets by
    assigning each set probability proportional to its size, and then
    sample a line from the set uniformly. (Here, in the sets defining
    interior vertices to faces, each vertex has weight $1/2$ as two
    vertices define a single line.)
     
    The only non-trivial case is when the algorithm picks a vertex in
    uniform from the interior of a face. This can be done in
    $O( \log n)$ time using the precomputed data-structure of
    \lemref{count:and:sample} for this face. Once the vertex is
    chosen, we know the two planes that induce it, and thus the line
    that had been chosen.

    Once the line had been chosen, computing the vertices along it can
    be done in linear time by intersecting it with all the planes of
    $\PT$, and keeping only the vertices on the interval on the line
    that lies inside $\PT$.
\end{proof}

\begin{lemma}
    \lemlab{iteration}%
    Let $\delta \in (0,1)$ be a fixed constant, let $\PT$ be a
    polytope in three dimensions with $t = O( \log n)$ faces, and let
    $\HSet$ be a set of $n$ planes, where all the faces of $\PT$ lie
    on planes of $\HSet$. Let $\VV = \VVX{\HSet} \cap \PT$ be the set
    of vertices of $\ArrX{\HSet}$ that lie inside $\PT$. One can
    compute a polytope $\PT'$, that is the intersection of $\PT$ with
    $O( \log n)$ halfspaces, such that $\PT'$ contains at most
    $|\VV|/n^{\delta}$ vertices of $\ArrX{\HSet}$.  The algorithm runs
    in $O( n^{1+\delta} \log^3 n \log \log n )$ time, and uses
    $O( \log n)$ separation oracle queries.
\end{lemma}

\begin{proof}
    Let $\eps = 1/n^\delta$. An $\eps$-net of $\VV$ for polytopes with
    $O( \log^2 n)$ faces, requires a sample of size
    \begin{equation*}
        m%
        =%
        O \pth{ \frac{\log^2 n \log \log n}{\eps} \log n }
        =%
        O(n^\delta \log^3 n \log \log n)
    \end{equation*}
    (this sample works with high probability). We invoke the
    over-sampling algorithm of \lemref{over:sample}, $m$ times for
    $k=O(\log^2 n)$. This takes
    $O( n k \log n + mn ) = O( n^{1+\delta} \log^3 n \log \log n )$
    time, and results in a set $Y$ of $O( n m) \approx n^{1+\delta}$
    points, that is super-set of an $\eps$-net (i.e., with high
    probability the set contains an $\eps$-net) for polytopes with
    $O(\log^2 n )$ faces.
    
    Next, we invoke the algorithm of \lemref{p:s:reduction}, to
    compute a polytope with $O( \log n)$ faces that does not contain
    any member of $Y$. Let $\PT'$ be the intersection of $\PT$ with
    this polytope. As $\PT'$ has $O( \log^2 n)$ faces, it follows by
    the $\eps$-net theorem that $\PT'$ contains at most
    $\eps \cardin{\VV}$ vertices of $\ArrX{\HSet}$ in it with high
    probability, as desired.
\end{proof}

Starting with $\Re^3$ as the initial polytope, the algorithm
repeatedly uses \lemref{iteration} to reduce the number of vertices of
the arrangement inside the current polytope by a factor of
$1/n^{\delta}$. In the $i$\th iteration, the current polytope has
$O( i \log n)$ faces, and as such the final polytope has at most
$O( (3/\delta) \log n)$ faces, as the algorithm has no vertices in it
after $\ceil{3/\delta}$ iterations. We conclude the following.

\begin{theorem}
    \thmlab{ulp:3:d}%
    For any $\delta \in (0,1)$, an instance of undecided \LP in three
    dimensions with $n$ constraints, can be solved using
    $O( (\log n)/\delta)$ separation oracle queries, in
    $O( n^{1+\delta} \log^4 n \log \log n)$ time.
\end{theorem}

\begin{remark}
    \remlab{ulp:d:b:3}%
    (A) Combining the algorithm of \lemref{cutting}, together with the
    algorithm of \thmref{ulp:3:d}, when the dimension is three,
    results in an algorithm that solves \ULP in $d>3$ dimensions, with
    $O( \delta^{-1} \log^{d-2} n)$ separation queries, and running
    time $\tldO(n^{1+\delta})$.

    (B) \thmref{ulp:3:d} can be further improved by reducing the
    complexity of the active polytope after every iteration. This only
    improves the running time by a polylogarithmic factor, and we omit
    the details for the sake of simplicity of exposition.
\end{remark}

\section{Covering points by monochromatic %
   balls using proximity queries}

\subsection{Learning a single monochromatic %
   ball using \NN/\FN queries}
\seclab{single:ball}

\paragraph*{Problem statement.}

The input is a set $\PS=\{\pp_1,...,\pp_n\}$ of $n$ points in
$\Re^d$. The points are either blue or red, but their color is not
initially provided.  We have access to a nearest-neighbor (\NN)
oracle, such that given a query point and a color, it returns the
closest point of this color to the query point. Similarly, we are
given access to a furthest-neighbor (\FN) oracle, that returns the
furthest point of this color in $\PS$.

The task at hand is to correctly classify all the given points as
either red or blue, minimizing the number of queries used.

\paragraph*{Single ball.}

Here, the assumption is that all the red points in $\PS$ are inside a
single ball, while all the blue points are outside. Our purpose here
is to develop an efficient algorithm that performs as few oracle
queries as possible, and decodes the color of all the points of $\PS$.

\begin{figure}
    \phantom{}%
    \hfill%
    \includegraphics{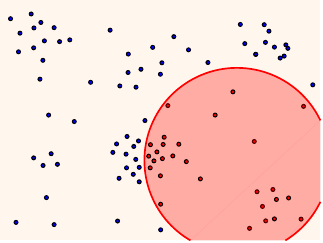}%
    \hfill%
    \includegraphics[page=2]{figs/one_disk}%
    \hfill%
    \phantom{}%
    \caption{A set of red and blue points in $\Re^2$ where a single
       disk suffices to separate the labels. A red furthest-neighbor
       (\FN) query and a blue nearest-neighbor (\NN) query at $q$
       confirms that a monochromatic disk centered at that point
       cannot contain every red point of the input.}
    \figlab{one_disk}
\end{figure}

\subsubsection{Lifting to three dimensions}
\seclab{lifting}

Let $\diskY{\pp}{r}$ denote the disk of radius $r$ centered at
$\pp$. Consider the mapping of such a disk to the point
\begin{equation*}
    \pd%
    = 
    \mDiskY{  \diskY{\pp}{r} } = \bigl(2\pp_x, 2\pp_y, r^2 - \pp_x^2 
    - \pp_y^2 \bigr).
\end{equation*}
This can be interpreted as a somewhat bizarre encoding of disks as
points in three dimensions.  We also map a point $\pq\in\Re^2$ in the
plane, to the plane
\begin{equation*}
    \mToPlaneX{ \pq} = \pth{ z = - \pq_xx - \pq_y{} y +
       \pq_x^2 + \pq_y^2}.      
\end{equation*}
Note that if $\pd$ is above $\hplane = \mToPlaneX{\pq}$ then
\begin{align*}
  &\pd_z \geq  -\pq_x\pd_x - \pq_y\pd_y + \pq_x^2 + \pq_y^2.      
    \quad\iff\quad%
    r^2 - \pp_x^2 -\pp_y^2  \geq - 2\pq_x\pp_x - 2\pq_y \pp_y +
    \pq_x^2 +
    \pq_y^2.      
  \\
  &\iff%
    r^2 \geq \bigl( \pp_x -  \pq_x\bigr)^2
    +\bigl( \pp_y -  \pq_y\bigr)^2
    \quad\iff\quad%
    \pq \in \diskY{\pp}{r}.
\end{align*}

\subsubsection{The algorithm}

The following algorithm is described in the plane, but also works in
higher dimensions with minor modifications.
The above lifting of disks to points (in three dimensions), and points
to planes, has the property that a point is above a plane $\iff$ the
original disk contains the original point. In particular, if a lifted
disk $\pd \in \Re^3$ is strictly below a plane $\hplane$, then the
original disk does not contain $\pq$ (i.e., the original point lifted
to $\hplane$).

In the lifted space, the input is a set of $n$ planes in three
dimensions. If a plane $\hplane$ is red, then the computed disk must
contain the original point, which means that the encoded disk $\pd$
must lie above $\hplane$. Namely, every red point, corresponds to a
commitment of the corresponding plane, to the halfspace lying
(vertically) above it. Similarly, an original blue point corresponds
to a commitment to the downward halfspace. We want to find a point in
this space which is feasible (after all the constraints have been
committed).

\paragraph*{The \ULP oracle queries.} %
A labeling oracle query on a plane, corresponds to providing the color
of the original point, which can be done using a single \NN colored
query. A feasibility oracle query, is a point $\pd$ in three
dimensions, which corresponds to a disk, which asks whether it
contains all the red points, and no blue points. The later can be
answered by performing a blue \NN query, and a red \FN query, and then
making a decision according to how the points interact with the query
disk.

\bigskip

As such, we can plug the lifted instance into the algorithm of
\thmref{u:d:lp:cutting}.  This algorithm also works verbatim in higher
dimensions. We thus get the following.
\begin{theorem}
    Let $\PS$ be a set of $n$ points in $\Re^d$.  Assume that there is
    an underlying coloring of the point set by (say) red and blue, and
    there is an oracle that can answer nearest-neighbor and
    furthest-neighbor colored queries on $\PS$.  The above algorithm
    computes a ball that contains all of the red points, and no blue
    points, if such a ball exists, using $O( \log^{d-1} n)$ \NN/\FN
    oracle queries. The running time of the algorithm is $O_d(n)$.
\end{theorem}

\subsection{Learning a cover by $\kopt$ %
   monochromatic balls using \NN queries}
\seclab{multi:ball}

\paragraph*{Problem statement.}

The input is a set of $n$ colored (say, by two colors) points $\PS$,
and assume that the points of $\PS$ can be covered by $\kopt$ balls,
such that each ball covers only points of a single color. Here, the
access to the points is via a \NN oracle queries (i.e., no \FN
queries). The task at hand is to classify the points correctly (i.e.,
decide their color) using a small number of oracle queries, by an
efficient algorithm.

\subsubsection{The algorithm}

\begin{lemma}
    \lemlab{canonical}%
    Let $\QS$ be a set of $m$ points in $\Re^d$, and let $\RangeSet$
    be the (infinite) set of all balls in $\Re^d$. One can compute, in
    $O(m^{d+2})$ time, the family of $O(m^{d+1})$ canonical sets
    induced by $\RangeSet$ on $\QS$, that is
    \begin{math}
        \FamilyX{\QS} = \Set{ \ball \cap \QS}{\ball \in \RangeSet}.
    \end{math}
\end{lemma}
\begin{proof}
    Using the lifting from \secref{lifting}, we have a set of $m$
    hyperplanes in $\Re^{d+1}$ dimensions. Each face of the
    arrangement of these hyperplanes corresponds to a canonical
    set. This arrangement can be computed in $O(m^{d+1})$ time. This
    also computes for each face (of any dimension) the canonical set
    that it realizes. Since each such canonical set is of size $O(m)$,
    the explicit listing of these sets requires $\Theta(m^{d+2})$
    time/space.
\end{proof}

For $i>0$, let $\PS_i$ denote the set of unlabeled points at the
beginning of the $i$\th iteration and let $n_i=|\PS_i|$ (i.e.,
$\PS_1 =\PS$, and $n_1=n$). The algorithm computes a random sample
$\Sample_i\subseteq \PS_i$ of size $O\pth{\kopt \log \kopt}$, and
determines the color of the points in $\Sample_i$ using \NN queries.
The algorithm then computes the set of canonical sets
$\Family_i =\FamilyX{\Sample_i}$, using \lemref{canonical}. For every
range $\range \in \Family_i$, such that
\begin{equation*}
    \sMeasureY{\range}{\Sample_i}%
    =%
    \frac{\cardin{\range \cap \Sample_i}}{ \cardin{\Sample_i}}%
    \geq%
    \frac{1}{2\kopt},    
\end{equation*}
and such that all the points in $\range$ are of the same color, the
algorithm runs a subroutine, described below in \secref{mono:ball}, to
decide if there is a monochromatic ball that contains all the points
of $\range$. Formally, the subroutine decides if there is ball
$\ball$, such that all the points of $\PS \cap \ball_i$ are colored by
the same color, and $\range \subseteq \ball$.  If no such ball is
found, the algorithm repeats this iteration until success.

The algorithm sets $\ball_i$ to be the ball computed, such that
$\cardin{\ball_i \cap \PS_i}$ is maximized among all such balls.  The
algorithm then adds $\ball_i$ to the computed cover, assigns all the
points in $\ball_i \cap \PS_i$ their color, and sets
$\PS_{i+1} \leftarrow \PS_i \setminus \ball_i$.

The algorithm stops once all points have been assigned their correct
color (i.e., $\PS_i =\emptyset$).

\subsubsection{Searching for a monochromatic %
   ball containing a set $\range$} %
\seclab{mono:ball}

The subroutine ,given a set of points $\range$ that are all (say) red,
searches for a ball $\ball$ such that the points in $\PS \cap \ball$
are all red, and $\range \subseteq \ball$. The subroutine is similar
in spirit to the single ball case of \secref{single:ball}, but the
details are somewhat different.

Specifically, consider the set $\BSet$ of all blue points in $\PS$
(this set is not explicitly known, as there are points that their
color is yet unknown), and consider the problem of computing a ball
that contains all the points of $\range$ and none of the points of
$\BSet$.  This is an implicit undecided optimization problem, which
via the lifting of \secref{lifting}, reduces to implicit undecided
\LP. A separation oracle here, in the original settings, is a query
ball $\qball$. If $\qball$ contains a blue point, one can find it by
performing a colored (i.e., blue) nearest-neighbor query on the set
original set of points $\PS$. Similarly, one can verify that $\ball$
contains all the red points of $\range$ by (say) scanning. Thus, one
can use the implicit undecided \LP algorithm of \lemref{cutting} (see
\remref{implicit}).

\subsubsection{Analysis}

Informally, each successful iteration reveals the color of a
$\Omega(1/\kopt)$-fraction of the unlabeled points. Specifically, at
the $i$\th iteration, at least one of the $\kopt$ balls of the optimal
solution must contain at least $n_i/\kopt$ points of $\PS_i$, which
are all of the same color (and are yet unlabeled).  As such, after
$O( \kopt \log n)$ iterations, the algorithm correctly exposes the
colors of the points in $\PS$.

\paragraph*{Preliminaries.}

We need the following standard results and definitions
\cite{h-gaa-11}.

\begin{defn}
    \deflab{measure}%
    Let $\RangeSpace = (\GroundSet,\RangeSet)$ be a range space, and
    let $\FGroundSet$ be a finite (fixed) subset of $\GroundSet$.
    Consider a subset $\SetC \subseteq \FGroundSet$ (which might be a
    multi-set). For a range $\range \in \RangeSet$, its
    \emphi{measure} is denoted by $\MeasureX{\range}$, and its
    \emphi{estimate} is $\sMeasureX{\range}$, where
    \begin{equation*}
        \displaystyle%
        \MeasureX{\range}%
        =%
        \frac{\cardin{\range \cap \FGroundSet}}{\cardin{\FGroundSet}}
        \qquad\,\text{~and~}\,\qquad%
        \sMeasureX{\range}%
        =%
        \frac{\cardin{\range \cap \SetC}}{\cardin{\SetC}}.        
    \end{equation*}
\end{defn}

\begin{defn}
    Let $\RangeSpace = (\GroundSet, \RangeSet)$ be a range space, and
    let $\FGroundSet$ be a finite subset of $\GroundSet$. For
    $0 \leq \eps, p \leq 1$, a subset $\SetC \subseteq \FGroundSet$ is
    a (relative) \emphi{$(\eps,p)$-approximation} for $\FGroundSet$,
    if for any range $\range \in \RangeSet$, we have
    \begin{align*}
      \MeasureX{\range} \geq p
      &\implies
        (1-\eps) \MeasureX{\range}
        \leq %
        \sMeasureX{\range}%
        \leq 
        (1+\eps) \MeasureX{\range}.\\
      \MeasureX{\range} < p
      &\implies
        \MeasureX{\range} - \eps p
        \leq %
        \sMeasureX{\range}%
        \leq 
        \MeasureX{\range} + \eps p.
    \end{align*}    
\end{defn}

\begin{theorem}[\cite{lls-ibscl-01,h-gaa-11}]
    \thmlab{relative}%
    Let $\RangeSpace = (\GroundSet, \RangeSet)$ be a range space with
    \VC dimension $\Dim$, and let $\FGroundSet$ be a finite subset of
    $\GroundSet$.  A sample $\SetC$ of size
    \begin{math}
        \displaystyle O \Bigl( \frac{1}{\eps^2p}\bigl( \Dim \log
        p^{-1} + \log {\BadProb^{-1}} \bigr)\Bigr)
    \end{math}
    is a relative $(\eps,p)$-\si{approx\-\si{imation}} with
    probability $\geq 1-\BadProb$.
\end{theorem}

\begin{fact}
    Let $\GroundSet=\Re^d$, and let $\RangeSet$ be the set of balls in
    $\Re^d$.  The $\VC$ dimension of $(\GroundSet, \RangeSet)$ is
    $d+1$.
\end{fact}

\begin{lemma}
    Given a set $\range \subseteq \PS$ of (say) red points, such that
    there exists a ball $\ball$ such that all the points of
    $\ball \cap \PS$ are of the same color, and
    $\range \subseteq \ball$, the subroutine of \secref{mono:ball}
    returns a monochromatic ball that covers the points of $\range$.
\end{lemma}
\begin{proof}
    Clearly, the (implicit) undecided \LP instance being created by
    the subroutine is feasible, and as described, the \NN queries on
    the blue points provide a separation oracle for this instance. As
    such, the undecided \LP would return a feasible point, which
    corresponds to the desired ball.
\end{proof}

\begin{lemma}
    \lemlab{2:5}%
    Under the assumption that the input can be covered by $\kopt$
    monochromatic balls, an iteration of the above algorithm succeeds
    with probability close to one, and computes a monochromatic ball
    that covers at least $2/(5\kopt)$ fraction of the uncolored
    points.
\end{lemma}
\begin{proof}
    In the beginning of the iteration there is a set
    $\QS \subseteq \PS$ of $n_i$ points that are not colored yet. As
    such, by assumption, there is a ball $\ball'$ in the optimal
    cover, which is monochromatic, and covers at least
    $\ceil{n_i/\kopt}$ points of $\QS$. Since the range space of balls
    in $\Re^d$ is of \VC dimension $d+1$, it follows that the sample
    $\Sample \subseteq \QS$ is a relative $(\eps,p)$-approximation for
    balls of $\QS$, where $p = 1/4\kopt$ and $\eps = 1/4$.  In
    particular, for this sample, with probability close to one, we
    have that
    $\sMeasureX{\ball'} \geq ( 1 - 1/4)\MeasureX{\ball'} \geq (1- 1/4)
    /\kopt$. This implies that
    $\sMeasureY{\ball'}{\Sample} \geq 1/(2\kopt)$.  As such, if the
    sample is successful, the algorithm would find a canonical set
    $\range$ of the sample that contains at least $1/(2\kopt)$
    fraction of the sample, and this canonical set has a monochromatic
    ball that contains the points of $\range$.
    
    As such, the subroutine would return a monochromatic ball $\ball$
    that contains $\range$. But then, by the relative approximation
    property of the sample, we have that
    \begin{equation*}
        \cardin{\ball \cap \QS}%
        = %
        \MeasureX{\ball} \cardin{\QS}
        \geq%
        \frac{1}{1+\eps}\sMeasureY{\range}{\Sample}\cardin{\QS}
        = %
        \frac{4}{5} \cdot \frac{1}{2\kopt}\cardin{\QS}
        \geq%
        \frac{2}{5\kopt} \cardin{\QS}.%
        \SoCGVer{\qedhere}
    \end{equation*}
\end{proof}

\begin{theorem}
    \thmlab{main}%
    Let $\PS$ be a set of $n$ points in $\Re^d$, such that there are
    (unknown) $\kopt$ monochromatic balls that cover all the points of
    $\PS$. Furthermore, assume we are given an oracle that can answer
    \NN colored queries on $\PS$. Then, one can compute the color of
    all the points of $\PS$ using
    \begin{math}
        O\pth{\kopt^{d+2} \log^{2d+3} n}
    \end{math}
    queries (this bound holds in expectation). The expected running
    time of the algorithm is
    \begin{math}
        O\bigl(\kopt^{d+2}  n \log^{d+1} \kopt \bigr).
    \end{math}
\end{theorem}
\begin{proof}
    By \lemref{2:5}, at each successful iteration the algorithm
    decreases the number of uncolored points by a factor of
    $1-2/(5\kopt)$. As such, after $m = 1 + \ceil{(5\kopt/2)\ln n}$
    iterations, we have that
    
    \begin{equation*}
        n_i%
        \leq%
        n(1-\frac{2}{5\kopt})^m%
        < %
        n\exp \Bigl( -\frac{2}{5\kopt} \cdot \frac{5\kopt}{2} \Bigr)^{\ln
           n}%
        =%
        1,
    \end{equation*}
    since $1-x \leq \exp(-x)$.

    Since the probability of a successful iteration is at least half,
    in expectation (and also with high probability), the overall
    number of iterations performed by the algorithm is
    $O( \kopt \log n)$.
       
    Every iteration, the algorithm takes a sample of size
    $\sSize = O( \kopt \log \kopt )$, and performs
    $O( \kopt \log \kopt)$ \NN queries to color these points. The
    algorithm spends $O(\sSize^{d+2})$ time generating all the
    canonical sets associated with the sample, and it potentially
    performs $O( \sSize^{d+1})$ calls to the subroutine of
    \secref{mono:ball}.  Each such call takes linear time, but uses
    $O( \log^{d+1} n)$ oracle \NN calls. As such, the total number of
    \NN queries performed by the algorithm is
    \begin{equation*}
        O\pth{ \bigl( \kopt \log \kopt + \sSize^{d+1}
           \log^{d+1} n\bigr)  \kopt \log n }
        =%
        O\pth{\kopt^{d+2} \log^{2d+3} n}. 
    \end{equation*}
    The running time for the first $O(\kopt)$ iterations is
    \begin{math}
        O\pth{ \bigl( \sSize^{d+2}+ \sSize^{d+1}
           n \bigr)  \kopt  }
        =%
        O\pth{\kopt^{d+2}  n \log^{d+1} \kopt}.         
    \end{math}
    Since the number of unlabeled points shrinks (in expectation) by a
    constant factor every $O(\kopt)$ iterations, we get that the
    overall running time is proportional to the above bound.
\end{proof}

\section{Covering by triangles, and terrain simplification}

\subsection{Learning a cover by $k$ monochromatic %
   triangles}
\seclab{multi:kgon}

\paragraph*{Problem statement.}

The input is a set $\PS$ of $n$ red and blue colored
points. Furthermore, assume that the points of $\PS$ can be covered by
$k$ monochromatic triangles (i.e., all points covered by a single
triangle have the same color). We assume access to the following two
oracles:
\begin{compactenumA}[leftmargin=0.8cm]
    \smallskip%
    \item \textsc{Colored triangle oracle}: Given a query triangle
    $\triangle$ and a color $c$, if the color of all the points of
    $\triangle \cap \PS$ is $c$, it return so.  Otherwise, the oracle
    returns a point in $\triangle \cap \PS$ with its color being
    different than $c$.

    \smallskip%
    \item \textsc{Sampling oracle}: Given a partial cover
    $\triangle_1, \ldots, \triangle_i$, the oracle return a point
    randomly sampled from $\PS \setminus \bigcup_i \triangle_i$. If
    all points are covered, the oracle reports this.
\end{compactenumA}

\paragraph{The algorithm.}  For $i>0$, let $\PS_i$ denote the set of
unlabeled points at the beginning of the $i$\th iteration and let
$n_i=|\PS_i|$ (i.e., $\PS_1 =\PS$, and $n_1=n$). The algorithm
computes a random sample $\Sample_i\subseteq \PS_i$ of size
$O\pth{k\log k}$, and determines the color of the points in
$\Sample_i$ using oracle queries (each query is a tiny triangle
containing a single point).

Let $\LS_i$ be the set of lines induced by pairs of points of
$\Sample_i$. Observe that
$\cardin{\LS_i} = O( \cardin{\Sample_i}^2) = O(k^2 \log^2 k)$.  The
arrangement of $\ArrX{\LS_i}$ has
$O(\cardin{\LS_i}^2) = O( k^4 \log^4 k)$ vertices.  The algorithm then
iterates over $3$-tuples of vertices of $\ArrX{\LS_i}$. There are
$O\pth{k^{12} \log^{12} k}$ triangles defined by such tuples. For each
such triangle $\triangle$, the algorithm first verifies that it is
sufficiently heavy in the sample (i.e.,
$\sMeasureY{\triangle}{\Sample_i} \geq 1/(10k)$), and that the points
it covers in the sample are monochromatic. Then, the algorithm issues
a query to the oracle to verify that all the points covered by the
triangle in $\PS$ are of the same color. If so, the algorithm adds
$\triangle$ to the cover, and sets
$\PS_{i+1} \leftarrow \PS_i \setminus \triangle$. The algorithm then
continues to the next iteration.

If no good triangle is found, the algorithm retries by setting
$\PS_{i+1} \leftarrow \PS_i$, and continuing to the next iteration.
The algorithm stops once all points have been assigned their correct
color (i.e., $\PS_i =\emptyset$).

\subsubsection{Analysis}

Consider a sequence $\PS = \permut{ \pp_1, \ldots, \pp_{dk}}$ in
$\Re^d$. Its \emph{average} is the point
$\pc = \avgX{\PS} = \sum_i \pp_i / dk$. For
$i \in \IRX{k} = \{ 1, \ldots, k\}$, let $\hplane^+_i $ be the
halfspace whose boundary passes through the points
$\pp_{(i-1)k +1}, \ldots, \pp_{ik}$, and it contains $\pc$ in its
interior. The \emphi{polytope} induced by $\PS$ is
$\PTX{\PS} = \bigcap_i \hplane^+_i$.

\begin{lemma}
    \lemlab{2k-gon}%
    Let $\PS \subseteq \Re^d$ be a set of points, and let $\PT$ be a
    closed polytope in $\Re^d$ formed by the intersection of $k$
    halfspaces. Furthermore, assume that $\PS$ is contained in the
    interior of $\PT$.  Then, there is a sequence $\PSA$ of
    $k(d^2 -d+1)$ points of $\PS$, such that
    $\PS \subseteq \PTX{\PSA} \subseteq \PT$.
\end{lemma}

\begin{proof}
    The proof is illustrated in \figref{self:defined} (A)--(F).  Let
    $\Body = \CHX{\PS}$ be the convex-hull of $\PS$. Every
    $(d-1)$-dimensional face of $\Body$ has a hyperplane $\hspA$ that
    supports the face. The intersection of these halfspaces is
    $\Body$. Let $\hspA^+_1, \ldots \hspA^+_t$ be these (closed)
    halfspaces -- that is $\Body = \bigcap_i \hspA^+_i$.

    \begin{figure}[h]
        \begin{tabular}{*{4}{c}}
          \includegraphics[page=1]{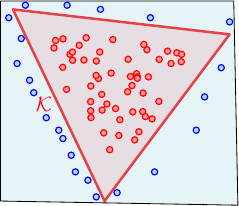}
          &%
            \includegraphics[page=2]{figs/self_defined}
          &%
            \includegraphics[page=3]{figs/self_defined}
          &%
            \includegraphics[page=4]{figs/self_defined}
          \\%
          (A) & (B) & (C) & (D) \\
          \includegraphics[page=5]{figs/self_defined}
          &%
            \includegraphics[page=6]{figs/self_defined}
          &%
            \includegraphics[page=7]{figs/self_defined}
          &%
            \includegraphics[page=8]{figs/self_defined}\\
          (E) & (F) & (G) & (H) 
        \end{tabular}
        \caption{Given a triangle $\PT$ that covers a set $\PS$ of $m$
           ``desired'' points, there is a triangle defined by nine
           points of $\PS$, such that this triangle covers at least
           $m/4$ of $\PS$, and is strictly contained inside the
           original triangle. }
        \figlab{self:defined}
    \end{figure}

    Let $\PT = \bigcap_{i=1}^k \hsp^+_i$, where $\hsp^+_i$ is a closed
    halfspace. Furthermore, let $\hsp_i^-$ be the complement (closed)
    halfspace to $\hsp^+_i$, for all $i$. Since
    $\hsp_i^- \cap \Body = \emptyset$, it follows by Helly's Theorem
    that there are $d$ halfspaces, say, $\hspA^+_1, \ldots, \hspA^+_d$
    such that $\hsp_i^- \cap \bigcap_{i=1}^d \hspA^+_i =
    \emptyset$. Let $\PS \cap \hspA_i$ be the set of $d$ points
    supporting this halfspace. We have that the sequence of $d^2$
    points
    $\PS \cap \hspA_1,\PS \cap \hspA_2, \ldots, \PS \cap \hspA_d$
    induces the cone $\bigcap_{i=1}^d \hspA^+_i$. In particular,
    compute the inducing sequences for each face of
    $\PT$. Concatenating these sequences, result in a sequence $\PSA$
    of $k d^2$ points, such that the induced polytope is an
    intersection of cones that avoids all the faces of $\PT$. That is,
    $\PS \subseteq \PTX{\PSA} \subseteq \PT$, as desired.

    The number of points in the sequence can be further reduced, by
    taking $v_i \in \PS$ to the closest vertex of $\Body$ to
    $\hspA_i^-$. Then there are $d$ halfspaces, whose boundary passes
    through $v_i$, such that their intersection avoids
    $\hspA_i^-$. Since they all share the common vertex $v_i$, the
    inducing sequence in this case is of length $d(d-1)+1$, or
    $k(d(d-1)+1)$ overall (the encoding of the inducing sequence is
    slightly different, but this is a minor issue which we ignore).
\end{proof}

\begin{lemma}
    \lemlab{1:10}%
    If the input can be covered by $k$ monochromatic triangles, then
    an iteration of the above algorithm succeeds with probability
    close to one, and computes a monochromatic triangle that covers at
    least $1/(10k)$ fraction of the unlabeled points.
\end{lemma}

\begin{proof}	
    In the beginning of the $i$\th iteration there is a set
    $\PS_i \subseteq \PS$ of $n_i$ points that are uncolored. As such,
    by assumption, there is a triangle $\triangle_i$ in the optimal
    cover which is monochromatic, and covers at least
    $\ceil{n_i/\kopt}$ points of $\PS_i$.
	
    \lemref{2k-gon} implies that there is a sequence $\PSA_i$ of $9$
    points of $\Sample_i \cap \triangle_i$, such that the hexagon it
    induces, denoted by $\PTX{\PSA_i}$, has the property that
    $\Sample_i \cap \triangle_i \subseteq \PTX{\PSA_i} \subseteq
    \triangle_i$. This hexagon can be triangulated by a bottom vertex
    triangulation into 4 triangles, each of which is induced by a
    triplet of vertices of $\ArrX{\LS_i}$ and is thus examined by the
    algorithm. In particular, one of these 4 triangles must contain at
    least $1/4\kopt$-fraction of the points of $\PS_i$, and all the
    points of $\PS$ (and thus of $\PS_i$) that this triangle covers
    are of the same color.  See \figref{self:defined}.

    By \thmref{relative}, we have that for a suitable choice of
    constants, $\Sample_i$ is a
    $\pth{ \frac{1}{4}, \frac{1}{16k}}$-relative approximation of the
    range space of triangles over the set $\PS_i$.  Due to that, the
    heaviest of the 4 triangles contains enough sample points to
    trigger a colored triangle oracle query by the algorithm, and will
    thus be added to the solution, as all the points in $\PS$ it
    covers are of the same color.  Thus, as in \lemref{2:5}, the
    ``heavy'' monochromatic triangle returned by the algorithm in this
    iteration covers $\Omega( n_i / k )$ points of $\PS_i$.
\end{proof}

\begin{lemma}
    The above algorithm has the following performance guarantees with
    high probability:
    \begin{compactenumA}[leftmargin=1.2cm]
        \item It performs $O( k \log n)$ iterations.

        \item Its overall running time is
        $O\pth{k^{14} \log^{13} k \log n}$
        
        \item It performs $O( k^2 \log k \log n)$ sampling oracle
        queries.

        \item It performs $O( k^{13} \log^{12} k\log n)$ colored
        triangle oracle queries.
    \end{compactenumA}
\end{lemma}
\begin{proof}
    By \lemref{1:10} each iteration of the algorithm, with probability
    close to one, labels $\Omega(1/k)$ fraction of the remaining
    uncolored points. It follows that with high probability (and in
    expectation) the algorithm is done after $O( k\log n )$
    iterations.

    The number of triangles tested in each iteration is
    $\alpha_1 = O\pth{k^{12} \log^{12} k}$. For each such triangle,
    verifying that it is heavy and monochromatic, for the sample,
    takes $\alpha_2 = O( k \log k )$ time. The algorithm issues
    potentially two colored triangle queries for such triangle. As the
    algorithm performs $\alpha_3 = O( k \log n)$ iterations, we
    conclude that the algorithm performs
    $\alpha_3 \alpha_1 = O( k^{13} \log^{12} k\log n)$ colored
    triangle oracle queries, and
    $O(\cardin{\Sample_i} \alpha_3) = O( k^2 \log k \log n)$ sampling
    oracle queries. The running time of the algorithm is
    $O(\alpha_1 \alpha_2 \alpha_3) = O\pth{k^{14} \log^{13} k \log
       n}$.
\end{proof}

\begin{theorem}
    \thmlab{cover:by:triangles}%
    {(A)} The input is a set $\PS \subseteq \Re^2$ of $n$ unlabeled
    points that can be covered by $k$ monochromatic triangles, such
    that the access to $\PS$ is only via a colored triangle and
    sampling oracles.  One can compute a cover of $\PS$ by
    $O( k \log n)$ triangles, with the following bounds (which hold
    with high probability):
    \begin{compactenumI}[leftmargin=1.2cm]
        \item It performs $O( k \log n)$ iterations.

        \item Its overall running time is
        $O\pth{k^{14} \log^{13} k \log n}$
        
        \item It performs $O( k^2 \log k \log n)$ sampling oracle
        queries.

        \item It performs $O( k^{13} \log^{12} k\log n)$ colored
        triangle oracle queries.
    \end{compactenumI}
    \medskip%

    \noindent%
    {(B)} If the input is an explicitly colored set $\PS$ of
    $n$ points in the plane, that can be covered by $k$ monochromatic
    triangles, then one can compute a cover of it by $O( k \log n)$
    triangles in $O( n k^{13}\log^{13} k\log n)$ expected time.
\end{theorem}
\begin{proof}
    Part {(A)} follows from the above.  As for part {(B)}, we
    implement the two oracles \naively, by scanning the colored input
    and carrying out the desired task. Since each oracle query can be
    done in $O(n )$ time, the above bound follows.
\end{proof}

\subsection{Terrain simplification}
\seclab{terrain}

\paragraph*{The problem.} %
Given a set $\PS$ of $n$ points in $\Re^3$ sampled from an unknown
function $f(x,y)$, find a piecewise-linear function $g(x,y)$, such
that for all $\pp=(x,y,f(x,y))\in \PS$, we have
$|f(x,y)-g(x,y)| \leq \epsH$, for some prespecified $\epsH>0$.

In the following, we assume that there exists a piecewise linear
function made out of $\kopt$ triangles that provides the desired
approximation, that is -- the point set has a
\emphi{$(\kopt,\epsH)$-terrain}.

\paragraph{Notations.} %
For a point $\pp \in \Re^3$, let $\prjX{\pp}$ be its projection to the
$xy$-plane. Similarly, for a set $X \subseteq \Re^3$, let
$\prjX{X} = \Set{ \prjX{\pp}}{\pp \in X}$.

\begin{defn}
    Let $\PS\subseteq \Re^3$ be a set of points.  A triangle
    $\triangle' \subseteq \Re^3$ is \emphi{$\epsH$-admissible} if all
    the points of $\PS$ that are vertically above/below it, are in
    vertical distance at most $\epsH$ from it.
\end{defn}

\begin{defn}
    A triangle
    $\triangle \subseteq \Re^2$ \emphi{covers} a set
    $S \subseteq \PS$, if
    \begin{compactenumi}[leftmargin=1cm]
        \item there exists a lifted three-dimensional triangle
        $\triangle' \subseteq \Re^3$, such that
        $\triangle = \prjX{\triangle'}$,

        \item $\prjX{S} \subseteq \triangle$,

        \item and $\triangle'$ is $\epsH$-admissible.
    \end{compactenumi}
\end{defn}

\begin{defn}
    A set $\TriSet$ of triangles in $\Re^2$ is a \emphi{$\epsH$-cover}
    of $\PS$, if every point of $\PS$ is covered by some triangle of
    $\TriSet$.
\end{defn}

Given a $\epsH$-cover of $\PS$ by $t$ triangles, it can be made into a
terrain by triangulating the arrangement created by the triangles, and
lifting the resulting triangles back to three dimensions. The
resulting terrain has $O(t^2)$ triangles.

\paragraph{Oracle access:} %
In the following we assume we have two oracles available:
\begin{compactenumI}
    \item \textsc{Validate triangle:} Check if a triangle is
    admissible.  Given a triangle $\triangle \subseteq\Re^2$ and
    $\epsH$, decide if the points of $\PS$ with their projection
    contained in $\triangle$ can be $\epsH$-approximated by a lifting
    of $\triangle$ to $\Re^3$. This oracle can be implemented via
    linear programming in linear time.

    \item \textsc{Sample uncovered point.} Given a collection of
    ($\epsH$ valid) triangles $\triangle_1, \ldots, \triangle_t$,
    return a random point of $\PS$ that is not $\epsH$-covered by
    these triangles.
\end{compactenumI}

\begin{theorem}
    \thmlab{t:simp}%
    Let $\PS$ be a set of $n$ points in $\Re^3$ that has a
    $(\kopt,\epsH)$-terrain (both $\kopt$ and $\epsH$ are provided),
    where the access to data is provided by the above two oracles.
    One can compute a $\epsH$-cover of $\PS$ with $O( \kopt \log n)$
    triangles. Alternatively, one can compute a
    $\bigl(O(\kopt^{2}\log^2 n), \epsH)$-terrain for $\PS$. The
    running time of the algorithm is
    $O\pth{\kopt^{14} \log^{13} \kopt \log n}$ (which also bounds the
    number of oracle queries performed).

    Without oracle access, the algorithm can be implemented to run in
    $O\pth{n \kopt^{14} \log^{13} k \log n}$ time.
\end{theorem}

\begin{proof}
    Surprisingly, the same algorithm as \thmref{cover:by:triangles}
    applies readily. The check whether or not a triangle is admissible
    (i.e., in the previous settings, this was checking if all the
    points in the triangle are all red or are all blue) is now reduced
    to a single validate triangle oracle query.

    Once the coverage by $k$ triangles is computed, one can compute
    the associated terrain by triangulating the arrangement of the
    projected triangles, and lifting the triangles back to three
    dimensions.

    Since the two oracles can be implemented in $O(n k \log n)$ time,
    as the cover has at most $O(k \log n)$ triangles, the bound on the
    running time follows, as each oracle call can be implemented in
    linear time, as described above.
\end{proof}

\paragraph*{Acknowledgments}

The authors thank Pankaj Agarwal for useful discussions.  We thank Lev
Reyzin for pointing out the work by Maass and \Turan
\cite{mt-albolgc-94, mt-hwtgl-94}.

\BibTexMode{%
   \SoCGVer{%
      \bibliographystyle{plain}%
   }%
   \NotSoCGVer{%
      \bibliographystyle{alpha}%
   }%
   \bibliography{undecided}%
}

\BibLatexMode{\printbibliography}

\end{document}